\documentclass[11pt]{amsart}
\usepackage{graphicx}
\usepackage{amscd}
\usepackage{amsmath}
\usepackage{amsfonts}
\usepackage{amssymb}
\usepackage{bbm}
\usepackage{setspace}
\usepackage{enumerate}         
\usepackage{color}
\usepackage{url}
\usepackage{amsthm}
\usepackage{hyperref}
\usepackage{bm}
\usepackage{xy}
\usepackage{color}
\usepackage{xcolor}
\usepackage{subfig}
\usepackage{bbm}
\usepackage{graphicx}
\usepackage{subfloat}
\usepackage{etoolbox}
\usepackage[pagewise]{lineno}
\usepackage{mathrsfs}
\usepackage[normalem]{ulem}
\usepackage{cancel} 
\usepackage{soul}
\allowdisplaybreaks[4]

\usepackage{natbib}

\usepackage{geometry}
\usepackage{verbatim}

\usepackage{tikz}
\usetikzlibrary{positioning,shapes,shadows,arrows}
\tikzstyle{abstract}=[rectangle, draw=black, rounded corners, fill=blue!40, drop shadow,
        text centered, anchor=north, text=white, text width=2.5cm]
\tikzstyle{comment}=[rectangle, draw=black, rounded corners, fill=green, drop shadow,
        text centered, anchor=north, text=white, text width=2cm]
\tikzstyle{myarrow}=[->, >=open triangle 90, thick]
\tikzstyle{line}=[-, thick]

\theoremstyle{plain}
\newtheorem{theorem}{Theorem}[section]

\newtheorem{corollary}[theorem]{Corollary}

\newtheorem{lemma}[theorem]{Lemma}

\newtheorem{problem}[theorem]{Problem}
\newtheorem{proposition}[theorem]{Proposition}

\newtheorem{definition}[theorem]{Definition}

\newtheorem{assumption}[theorem]{Assumption}

\theoremstyle{remark}
\newtheorem{remark}[theorem]{Remark}
\newtheorem{example}[theorem]{Example}
 
\theoremstyle{hp}
\newtheorem{hp}[theorem]{Hypothesis}

\numberwithin{equation}{section}

\newcommand{\rsto}{]\!\kern-1.8pt ]}
\newcommand{\lsto}{[\!\kern-1.7pt [}

\numberwithin{equation}{section}

\newcommand{\cadlag}{c\`adl\`ag\,}

\newcommand{\Ind}[1]{\mathrm{1}_{\left\{#1\right\}}}

\newcommand{\FF}{\mathbb{F}}
\newcommand{\GG}{\mathbb{G}}

\newcommand{\RR}{\mathbb{R}}
\newcommand{\QQ}{\mathbb{Q}}

\newcommand{\PP}{\mathbb{P}}

\newcommand{\EE}{\mathbb{E}}
\newcommand{\XX}{\mathbb{X}}

\newcommand{\cA}{\mathcal{A}}
\newcommand{\cB}{\mathcal{B}}

\newcommand{\cF}{\mathcal{F}}
\newcommand{\cG}{\mathcal{G}}
\newcommand{\cH}{\mathcal{H}}

\newcommand{\cM}{\mathcal{M}}
\newcommand{\cN}{\mathcal{N}}

\newcommand{\cP}{\mathcal{P}}

\newcommand{\cS}{\mathcal{S}}

\newcommand{\cV}{\mathcal{V}}

\newcommand{\Ex}[2]{\mathbb{E}^{#1}\left[#2\right]}                     
\newcommand{\Excond}[3]{\mathbb{E}^{#1}\left[\left.#2\right|#3\right]}  

\newcommand{\rc}{r^{c}} 

\newcommand{\rck}[1]{r^{c,#1}} 
\newcommand{\rfk}[1]{r^{f,#1}} 

\newcommand{\Bck}[1]{B^{c,#1}}
\newcommand{\Bfk}[1]{B^{f,#1}}

\newcommand{\BIk}[1]{B^{I,#1}}

\newcommand{\tI}{\mathcal{I}}

\renewcommand{\cite}{\citet}

\makeatletter
\patchcmd{\@maketitle}
  {\ifx\@empty\@dedicatory}
  {\ifx\@empty\@date \else {\vskip3ex \centering\footnotesize\@date\par\vskip1ex}\fi
   \ifx\@empty\@dedicatory}
  {}{}
\patchcmd{\@adminfootnotes}
  {\ifx\@empty\@date\else \@footnotetext{\@setdate}\fi}
  {}{}{}
\newcommand{\subjclassname@JEL}{JEL Classification}
\makeatother

\begin{document}

\title[BSDEs of xVA]{
A unified approach to xVA with CSA discounting and initial margin}


\author{Francesca Biagini}
\address[Francesca Biagini]{LMU M\"unchen, Mathematics Institute,\newline
\indent Theresienstr. 39, D-80333 Munich, Germany}
\email[Francesca Biagini]{biagini@math.lmu.de}%

\author{Alessandro Gnoatto}
\address[Alessandro Gnoatto]{University of Verona, Department of Economics, \newline
\indent via Cantarane 24, 37129 Verona, Italy}
\email[Alessandro Gnoatto]{alessandro.gnoatto@univr.it}

\author{Immacolata Oliva}
\address[Immacolata Oliva]{Sapienza University of Rome, Department of Methods and Models for Economics, Territory and Finance\newline%
\indent Via del Castro Laurenziano 9, 00161 Rome, Italy}
\email[Immacolata Oliva]{immacolata.oliva@uniroma1.it}%

\begin{abstract}
In this paper we extend the existing literature on xVA along three directions. First, we enhance current BSDE-based xVA frameworks to include initial margin \textcolor{black}{in presence of defaults}. Next, we solve the consistency problem that arises when the front-office desk of the bank \textcolor{black}{uses trade-specific discount curves (CSA discounting) which differ from the discount rate} adopted by the xVA desk. Finally, we 
\textcolor{black}{clarify the impact of aggregation of several sub-portfolios of trades on the xVA-valuation of the resulting global portfolio and study related non-linearity effects.}
\end{abstract}


\date{\today}

\maketitle

\section{Introduction}\label{intro} 

As a consequence of the 2007-2009 financial crisis, academics and practitioners are revisiting the
valuation of financial products in several aspects. In particular, the value of a product should account for the possibility of default of any agent involved in the transaction. Financial regulations, such as Basel III/IV and Emir, are also driving the methodological development.

All these issues are represented at the level of valuation equations by introducing value adjustments (xVA), which are further terms to be added or subtracted to an idealized reference price, computed in the absence of the aforementioned frictions, in order to obtain the final value of the transaction. 

In this paper we aim to provide a unified and comprehensive framework for pricing counterparty risk with xVA and assessing funding costs by using BSDEs techniques in a market described by diffusions. We then consider the consistency problem between xVA pricing 
equations and the CSA discounting rules, which originates from the quoting mechanism of market standard instruments, and solve it by relying on an invariance property of linear BSDEs.  
\textcolor{black}{
Furthermore, we extend our approach to include pricing of a global portfolio possibly stratified in an aggregation of several subsets of claims, and study non-linearity effects due to incremental xVA charges. 
}

\textcolor{black}{
In order to better specify our contribution, we first shortly recall the literature on xVA. Since the research on xVA is very wide, our overview is far from being exhaustive. }
The first contribution on the subject is \cite{duhu96}. Before the 2007--2009 financial crisis, the works of \cite{brigoMasetti} and \cite{cherubini05} analyze the concept of \emph{credit valuation adjustment} (CVA) for the first time. 
The possibility of default of both counterparties involved in the transaction, represented by the introduction of the \emph{debt valuation adjustment} (DVA), is investigated, among others, in \cite{bripapa11} and \cite{bricapa14}.

Apart from the issue of default risk, another important source of concern for practitioners and academics is represented by funding costs. 

A parallel stream of literature emerged during and after the financial crisis, to 
generalize valuation equations to account for features such as the presence of collateralization agreements. 
In a Black-Scholes economy, \cite{pit10} provides valuation formulas in presence or absence of collateral agreements. 
\cite{pit12} generalizes the issue in a multi-currency economy, see also \cite{fushita10b},  \cite{fushita09}, and \cite{gnoatto2020}. 
The \emph{funding valuation adjustment} (FVA) under several alternative assumptions on the so called \emph{Credit 
Support Annex} (CSA) is derived in \cite{papebri2011}, while \cite{bripa2014ccp} also discusses the role of central 
counterparties in the context of funding costs. A general approach to funding issues in a semimartingale setting is 
provided by \cite{BieRut15}.

Both funding and default risk need to be unified in a unique pricing framework. 
Contributions in this sense can be found in \cite{bbfpr2018} by means of the so-called \textit{discounting approach}. 
In a series of papers, Burgard and Kjaer generalize the classical Black-Scholes replication approach to include many effects, 
see \cite{bj2011a} and \cite{bj2013}. A more general BSDE approach is provided by \cite{crepey2015a}, \cite{crepey2015b}, 
\cite{BiCaStu2018} and \cite{BiCaStu2019}. The equivalence between the discounting and the BSDE-based replication approaches 
is demonstrated in \cite{bbfpr2018}.

The importance of xVA is reflected by the increasing number of monographs on the subject, see e.g. \cite{brimopa2013}. 
An advanced BSDE-based treatment is provided by \cite{crepeyBook14}. A detailed analysis of how to construct large hybrid 
models for counterparty risk simulations are provided in \cite{green2015}, \cite{listag2015} and \cite{sokol2014}, while 
\cite{gregory15} provides an accessible introduction to most aspects of the topic.

The present paper contributes to the literature on xVA along the following directions: on the theoretical side, 
we provide a unifying treatment of different aspects of xVA that have been separately treated by different authors \textcolor{black}{in} a comprehensive \textcolor{black}{approach}; on the other hand, we provide a \textcolor{black}{rigorous} mathematical treatment of practical issues arising in the context of xVA calculation and management in current market practice.

\textcolor{black}{More precisely,} we propose an xVA framework using BSDEs techniques in a market described by diffusions. Our BSDEs, introduced under a progressively enlarged filtration \textcolor{black}{$\mathbb{G},$} are specified up to a random time horizon, 
given by the minimum between the default time of the counterparty, the default time of the bank, and the maturity of the contract. For the sake of simplicity, we first discuss the well posedness of the BSDE for \textcolor{black}{a portfolio including a contingent claim.} \textcolor{black}{Our approach unifies several contributions from the existing literature. First, we include initial margins in the BSDE driver as in \cite{Gobet2018}, but we extend their setting by including defaults. Moreover, we consider different discounting rules without resorting to measure changes as in \cite{BiCaStu2018}. Finally, we provide a price decomposition in terms of a clean value and an adjustment process, by considering the associated pre-default BSDE with respect to the reference filtration $\mathbb{F}$ along the lines of \cite{crepey2015a}, \cite{crepey2015b} and \cite{BieRut15}.}

Concerning discounting rules, it is nowadays accepted that  the trading 
activity is funded by resorting on different sources of liquidity. More precisely,  the existence
of \textcolor{black}{different} funding curves leads to the consistency problem between xVA pricing equations
and the so-called CSA discounting rules. The latter originates from the quoting
mechanism of market standard instruments. Such instruments are quoted under the assumption that
they are perfectly collateralized. Since a perfectly collateralized transaction is funded by the collateral
provider, the discounting rate applied to evaluate market instruments is given by a collateral rate,
which typically corresponds to an overnight interest rate. The presence of multiple assumptions on the
collateral rate implies the co-existence of quotes with different discounting rates, which are in general
at odds with the discounting rate dictated by the xVA pricing BSDE. A further contribution of our work consists in a solution to the consistency issue by relying on an invariance property of linear BSDEs. \textcolor{black}{This is realized by introducing a further value adjustment term, called \textit{Discounting Value Adjustment} (DiscVA), representing the clean price adjustment compared to the one officially accepted by the front-office desk.} 

Moreover, the counterparty could be economically related to other entities, a typical example being given by a 
parent-subsidiary relation or in general by the existence of dedicated sub-portfolios. Such a scenario 
results in a multi-layer \emph{hierarchical} structure \textcolor{black}{in the resulting global portfolio}, which also takes into account the potentially different funding 
policies implemented by each of the aforementioned financial entities. We study the impact of such a stratified structure on the formulation of pricing equations.
In realistic situations, the global portfolio of positions between the bank and the counterparty is
in general an aggregation of several subsets of claims, \textcolor{black}{whose} structure is dictated by legal agreements in force between the two agents. We allow for the presence of multiple
agreements for the exchange of margins (margin sets) and multiple netting sets. We adapt our xVA
framework to arbitrary configurations of aggregation levels, while preserving the well posedness of the underlying BSDE.

Finally, we present incremental xVA charges for new potential trades under the proposed xVA framework: given the 
presence of portfolio effects in the computation of value adjustments, and given an existing portfolio of $K$ trades, 
the xVA charge for a new potential $(K+1)$-th trade is computed as the difference of the xVAs of two portfolios. 
\textcolor{black}{More precisely}, this corresponds to the difference between the xVA of the extended portfolio, consisting of $(K+1)$ 
trades, and the xVA charge of the base portfolio of $K$ trades. \textcolor{black}{In this way we are able} to describe non-linearity of portfolio-wide valuation. Existing discussions concerning the impact of the 
incremental cash flow of a trade on the balance sheet of the dealer can be found in \cite{bj2013}, \cite{castagna2013}, 
\cite{castagna2014}.

Given our focus on discounting and aggregation levels, in this paper we do not discuss \emph{capital valuation adjustment} 
(KVA). The issue is treated in recent papers such as \cite{alcacre2016}, \cite{alcacre2016b} and \cite{alcacre2017}, 
see also the ongoing discussion in \cite{ads2019}.
This is beyond the scope of the present paper and leave it for future research.

The paper is organized as follows. In Section \ref{Sect:preliminaries} we formalize in mathematical terms the main financial 
concepts related to the xVA framework. Section \ref{Sect:single} \textcolor{black}{first} describes the results related to the xVA 
evaluation when only one transaction is taken into account, \textcolor{black}{extending the analysis to CSA discounting}. 
\textcolor{black}{Furthermore, we study the xVA framework when the portfolio includes several contracts and investigate the role played by non-linearity}.  
Section~\ref{sec:example} provides an example and some numerical results illustrating \textcolor{black}{the results of Section \ref{Sect:single}.}
In Appendix \ref{app} we gather some results from the literature used to derive the main results.

\section{The financial setting}\label{Sect:preliminaries}
We fix a time horizon $T<\infty$ for the trading activity. We consider two agents named the \emph{bank} (B) and the 
\emph{counterparty} (C). Unless otherwise stated, throughout the paper we assume the bank's perspective 
and refer to the bank as the \emph{hedger.} 
\textcolor{black}{All processes are modeled over a probability space $(\Omega, \mathcal{G}, \mathbb{P})$. 
}
\textcolor{black}{
Let $W^\PP=\left(W^\PP_t\right)_{t\,\in\,[0,T]}$ be a $d$-dimensional Brownian motion on $(\Omega,\mathcal{G},\mathbb{P})$ with associated natural filtration $\mathbb{F}.$
}

We denote by $\tau^B,$ resp. by $\tau^C,$ the \emph{time of default} of the bank, resp. of the counterparty
and \textcolor{black}{by} $H^j_t:=\Ind{\tau^j\leq t}, \, j\,\in\, \{B,C\},$   
\textcolor{black}{the associated jump process.}

\textcolor{black}{
Default times \textcolor{black}{$\tau^B$ and $\tau^C$} are assumed to be exponentially distributed random variables with time-dependent intensity 
$$\Gamma^j_t=\int_0^t\textcolor{black}{\lambda^{j,\mathbb{P}}_s} ds, \, j\,\in\,\{B,C\}, \, \,t \,\in\, [0,T],$$ 
where $\lambda^j$ are non-negative bounded $\FF$-adapted processes. 
}
\textcolor{black}{
Let $\mathbb{H}^j = (\mathcal{H}_t^j)_{t \,\in\, [0,T]}, \, j \,\in\, \{B,C\},$ be the natural filtration of $H^B,\, H^C,$ respectively. 
}
\textcolor{black}{
On $(\Omega, \cG,\PP)$ we consider the filtration $\GG = \FF \vee \mathbb{H}.$ All filtrations are required to satisfy the usual hypotheses of completeness and right-continuity.} 

In the present paper we will extensively make use of the so called \emph{Immersion Hypothesis.}

\begin{hp} \label{hp:H}   
Any local $(\FF,\PP)$-martingale is a local $(\GG,\PP)$-martingale. 
\end{hp} 

\textcolor{black}{Note that, by Hypothesis \ref{hp:H}, $W^{\mathbb{P}}$ is also a $\GG$-Brownian motion. By \cite[Corollary 5.2.4]{BieRut2004} we obtain that the predictable representation property holds in $\GG$ with respect to $W, \, M^B, \, M^C,$ where 
$$M_t^{j} := H_t^j - \int_{0}^{t \wedge \tau^j} \lambda_u^j du, \; t \,\in\, [0.T],$$ 
for $j \,\in\, \{B,C\}.$} 

\textcolor{black}{From now on, we set 
	\begin{equation} \label{eq:tau}
		\tau := \tau^B \wedge \tau^C \wedge T. 
	\end{equation}
}

\subsection{Basic traded assets} 
\textcolor{black}{
We now introduce the market model.} 
For $d \geq 1,$ we denote by $S^i$, $i=1,\ldots, d$ the \emph{ex-dividend price} (i.e. the price) of risky \textcolor{black}{
assets} with associated 
\emph{cumulative dividend processes} $D^i$. All $S^i$ are assumed to be c\`adl\`ag $\FF$-semimartingales, while the cumulative dividend streams $D^i$ are $\FF$-adapted processes of finite variation with $D^i_0=0$.

We \textcolor{black}{consider} the following coefficient functions:
\begin{align} \label{coeff_def}
\begin{split}
\mu: & \left(\RR_{+}\times \RR^d, \mathcal{B}\left(\RR_{+}\times \RR^d\right)\right) \mapsto \left(\RR^d,\mathcal{B}\left(\RR^d\right)\right),\\
\sigma: & \left(\RR_{+}\times \RR^d, \mathcal{B}\left(\RR_{+}\times \RR^d\right)\right) \mapsto \left(\RR^{d\times d},\mathcal{B}\left(\RR^{d\times d}\right)\right),\\
\kappa: & \left(\RR_{+}\times \RR^d, \mathcal{B}\left(\RR_{+}\times \RR^d\right)\right) \mapsto \left(\RR^d\textcolor{black}{_+},\mathcal{B}\left(\RR^d\textcolor{black}{_+}\right)\right),
\end{split}
\end{align}
which are assumed to satisfy standard conditions ensuring existence and uniqueness of strong solutions of SDEs driven by the Brownian motion $W^\PP$. The matrix process $\sigma$ is 
invertible at every point in time. We assume that
\begin{align}\label{eq:asset_prox}
\begin{cases}
dS_t&=\mu(t,S_t)dt+\sigma(t,S_t)dW^\PP_t\\
S_0&=s_0\in\RR^{d}  
\end{cases}
\end{align}
on $[0,T].$ 
Note that we are not postulating that the processes $S^i$ are positive. The dividend processes $D^i$ are specified via 
\begin{equation} \label{eq:prox}
(D^1_t,\ldots,D^d_t)^\top = \int_0^t \kappa(u,S_u) du, \; t \,\in\, [0,T], 
\end{equation}
for $\kappa$ given in \eqref{coeff_def} such that $\int_0^T \textcolor{black}{|}\kappa(u,S_u)\textcolor{black}{|} du <\infty$ $\PP$-a.s. 

\textcolor{black}{The cumulative dividend price associated to the $i$-th asset is given by 
\begin{equation} \label{eq:S_cld} 
S^{i,cld}_t := S^i_t + B^{i}_t \int_{(0,t]}\frac{dD^i_u}{B^{i}_u}, \; \ i=1,\ldots d \;, \, t \,\in\,[0,T].
\end{equation}
}
We assume the existence of an indexed family of cash accounts $(B^x)_{x \,\in\,\tI},$ where the stochastic process 
$r^x := (r^x_t)_{t \geq 0}$ is bounded and $\FF$-adapted for all $x \,\in\, \tI.$ 
The set of indices $\tI$ embodies the type of agreement the counterparties establish in order to mitigate the 
counterparty credit risk. We will specify the characteristics of the aforementioned indices later on.  

All cash accounts, with unitary value at time $0$, are assumed to be strictly positive 
continuous processes of finite variation of the form 
\begin{align} \label{def:cash_account}
B^x_t := \exp\left\{\int_0^t r^{x}_s ds\right\}, \,t \,\in\, [0,T]. 
\end{align}

We introduce two risky bonds with maturity $T^\star \leq T$ and rate of return $\bar{r}^j + \lambda^j,$ \textcolor{black}{$j \,\in\, \{B,C\},$} 
issued by the bank and the counterparty, \textcolor{black}{respectively,} with dynamics

\begin{align}\label{eq:bond_prox}
dP^j_t=\left(\bar{r}^j_t+\lambda^j_t\right)P^j_tdt-P^j_{t-}dH^j_t, 
\, j\in\{B,C\}.
\end{align}

The payment stream of a financial contract is represented by an $\FF$-adapted c\`adl\`ag process of finite variation $A = (A_t)_{t \,\in\, [0,T]}$, 
as in \cite{crepey2015b}.  
We use the notation $\Delta A_t:=A_t-A_{t-}$ for the jumps of $A$. 
To include the presence of default events, we define the process 
$\bar{A} = \left(\bar{A}_t \right)_{t\,\in\,[0,T]}$ by setting
\begin{equation}\label{A_bar}
\bar{A}_t :=\Ind{t<\tau}A_t+\Ind{t\geq \tau}A_{\tau-} \;. 
\end{equation}
\textcolor{black}{Note that $A_{\tau-}$ represents the last payment before default, see also \cite{BriMor2018}.}

\subsection{\textcolor{black}{Repo-trading and collateralization}} \label{repo_trading}
In line with the existing literature, we assume that the trading activity on the risky assets is \emph{collateralized.} 
This means that borrowing and lending activities related to risky securities are financed via security lending or \emph{repo market,} 
\textcolor{black}{see} \cite{BiCaStu2018} 
Since transactions on the repo market are collateralized by the risky assets, repo rates are lower than unsecured funding rates. 
As argued in \cite{crepey2015a}, assuming that all assets are traded via repo markets is not restrictive. 
In case the transactions are fully collateralized, this translates in the following equality
\begin{align} \label{eq:repoConstraint}
\xi^i_tS^i_t + \psi^i_tB^i_t = 0, \; i = 1\,\ldots,d, \,t \,\in\, [0,T]\textcolor{black}{,}   
\end{align}
\textcolor{black}{where} $B^1,\ldots, B^d$ 
are the cash accounts associated to the risky assets $S^1,\ldots,S^d$. 

It is worth noting that $\xi_t^i, \, i = 1, \ldots, d,$ may be either positive or negative. Here $\xi_t^i >0$ means that we are in a long position, which has to be financed by collateralization. On the other hand, $\xi_t^i < 0$ implies that the $i$-th asset is shorted, so that the whole amount of collateral is deposited  in the riskless asset. 

Condition \eqref{eq:repoConstraint} plays an important role in precluding trivial arbitrage opportunities among different cash accounts. 

Within the bank, the trading desk borrows and lends money from/to the treasury desk. 
Borrowing and lending rates are allowed to differ, hence we denote by $\rfk{b},\rfk{l}$ the rate at which 
the trading desk borrows from and lends to the treasury desk, respectively. 
\textcolor{black}{In the same} notation \textcolor{black}{as} in \eqref{def:cash_account}, we introduce the associated cash accounts 
$\Bfk{b},\Bfk{l}.$ 
This means that if the position of the trading desk is negative, i.e. $\psi^f = \psi^{f,b}<0,$ the trading desk 
borrows from the treasury desk at the rate $r^{f,b}.$ Conversely, if the position of the trading desk is positive, 
i.e. $\psi^f = \psi^{f,l}>0,$ the trading desk lends money to the treasury desk with remuneration $r^{f,l}.$ 
\textcolor{black}{Since simultaneously borrowing and lending from the treasury desk is precluded, we set 
\begin{equation} \label{eq:bl_sim}
\psi^{f,l}_t\psi^{f,b}_t=0, \, \mbox{ for all } t \,\in\, [0,T]. 
\end{equation} }

\textcolor{black}{For what concerns financial contracts, collateralization is a method to minimize losses due to default 
of the counterparty by using margins.
}
In the financial jargon, a \emph{margin} represents an economic value, either in the form of cash or risky securities, 
exchanged between the counterparties of a financial transaction, in order to reduce their risk exposures. In line with the current market practice, we distinguish between \emph{initial margin} and \emph{collateral} 
(or \emph{variation margin}), that we present in \textcolor{black}{the following.}

\subsubsection{Variation margin}\label{subsect:VM}
A collateral is posted between the bank and the counterparty to mitigate counterparty risk. 
The collateral process $C=(C_t)_{t\,\in\,[0,T]}$ is assumed to be $\GG$-adapted. 
We follow the convention of \cite{BiCaStu2018}  and \cite{crepey2015a}: 

\begin{itemize}
\item If $C_t>0,$ we say that the bank is the \emph{collateral provider.} It means that the counterparty measures a positive exposure 
			towards the bank, so it is a potential lender to the bank, hence the bank provides/lends collateral to reduce its exposure. 

\item If $C_t<0,$ we say that the bank is the \emph{collateral taker.} It means that the bank measures a positive exposure towards 
			the counterparty, so it is a potential lender to the counterparty, hence the counterparty provides/lends collateral to reduce 
			its exposure. 
\end{itemize}

Let $\cV= \left(\cV_t\right)_{t\,\in\,[0,T]}$ be a generic $\GG$-adapted process, representing either the value of the trade including counterparty risk and funding adjustments or the clean value process, as it will be clarified later on. We assume that 
$C_t := f(\cV_t), \, t \,\in\, [0,T],$ where $f: \, \RR \,\rightarrow\, \RR$ is a Lipschitz function. This assumption allows to cover realistic collateral specifications, see e.g. \cite{listag2015} and \cite{BaFuMa2019}.

If there is a collateral agreement (or a multitude of agreements)  between the bank and the counterparty, in evaluating 
portfolio dynamics we need to make a distinction between the value of the portfolio and the wealth of the bank, the two 
concepts being distinguished since the bank is not the legal owner of the collateral (prior to default).

In this paper collateral is always posted in the form of cash, in line with standard \textcolor{black}{current market practice}. Moreover, we assume 
\emph{rehypothecation,} meaning that the holder of collateral can use the cash to finance her trading activity. 
This is the opposite of \emph{segregation,} where the received cash collateral must be kept in a separate account and 
can not be used to finance the purchase of assets. 

We associate the following interest rates to the collateral account:
\begin{itemize}
\item $\rck{l}$ with account $\Bck{l},$ representing the rate on the collateral amount \emph{received} by the bank 
			who posted collateral to the counterparty.
\item $\rck{b}$ with account $\Bck{b},$ representing the rate on the collateral amount \emph{paid} by the bank who 
			received collateral from the counterparty.
\end{itemize}

We simply set $\rc=\rck{l}=\rck{b}$ in case there is no bid-offer spread in the collateral rate. 
Possible choices for the collateral rate are e.g. EONIA for EUR trades, Fed Fund for USD and SONIA for GBP trades. 
Such rates are overnight rates with a negligible embedded risk component. The choice of such approximately risk-free rates
as collateral rates is motivated by market consensus. However, two counterparties might enter a collateral agreement 
that involves a remuneration of collateral at any other risky rate of their choice. 
Here we do not assume any requirements on collateral rates. This allows us to cover the quite common situation where the 
collateral rate agreed between the two counterparties in the CSA is defined by including a real valued 
spread over some market publicly observed rate, e.g. EONIA $- 50\, bps,$ where $bps$ stands for \emph{basis points.} 

For the collateral account we have the following equations: 

\begin{enumerate}
\item if $C_t>0$, 
			then the bank has lent $\psi^c_t=\psi^{c,l}_t < 0$ units of the collateral 
			account to the counterparty, i.e., 
			\begin{equation} \label{eq:coll_bank}
				\psi^{c,l}_t\Bck{l}_t=-C^+_t, \,t \,\in\, [0,T]; 
			\end{equation} 
			
\item if $C_t<0$, 
			then the bank has borrowed $\psi^c_t=\psi^{c,b}_t>0$ units of the collateral 
			account from the counterparty, i.e., 
			\begin{equation} \label{eq:coll_counterparty}
				\psi^{c,b}_t\Bck{b}_t=C^-_t, \,t \,\in\, [0,T]. 
			\end{equation}
\end{enumerate} 

It is worth noting that we use the following conventions: $x^+:=\max\{x,0\}$, $x^-:=\max\{-x,0\}$ so that $x=x^+-x^-.$ 
Note that this is in contrast to the convention adopted e.g. in \cite{bj2011b,bj2013}.

\subsubsection{Initial margin} \label{subsect:IM}
The collateralization represented by the variation margin is imperfect, due to the margin period of risk phenomenon: 
a defaulted counterparty stops posting collateral. However, bankruptcy procedure requires a certain time interval 
(typically 10 or 20 days) before the close-out payments are performed. 
This results in a period of time where the value of the transaction oscillates in the absence of an adjustment of 
the collateral account, hence producing an exposure. 
This is one of the reasons for the introduction of initial margins, which constitutes a further form of collateral.

According to the EMIR regulation, starting from 2020, most agents participating in an OTC transaction will be 
forced to post initial margin, which constitutes an additional form of collateral. 
Initial margin, according to \cite{thm2016} is a misnomer, as an initial margin is 
not only initial, but it is periodically updated during the lifetime of the trade. It is \emph{initial} in the sense 
that it is meant to provide a coverage from the initial point in time, where there is a default of the counterparty 
in a collateralized transaction. 

It is important to stress that, differently from variation margin, an initial margin can not be rehypothecated, 
but it is instead segregated. From the point of view of the wealth dynamics, this means that initial margin received 
from the counterparty can not be used by the trading desk as a component of the value of the portfolio. 
However, the received initial margin represents a loan from the counterparty that must be remunerated, hence funding 
costs related to initial margin will appear in the self-financing condition, see   \eqref{eq:sefFinancing} and \eqref{eq:dynamics} for further details.

\begin{remark} 
\textcolor{black}{
Initial margins are usually quantified by using suitable risk measures, 
such as \emph{value at risk} or \emph{expected shortfall,} as we will specify in \eqref{def:IM}. 
Expected shortfall is a popular choice to compute the initial margin for credit derivatives, since it is a 
coherent risk measure. 
Recently, the \emph{International Swaps and Derivatives Association} (ISDA) has proposed a novel methodology, 
the so called \emph{Standard Initial Margin Model} (SIMM), see \cite{isda_simm}. SIMM provides 
some standardized formulae to evaluate initial margin on non-cleared derivatives, based on using portfolio 
sensitivities instead of historical simulations. 
}
\end{remark} 

We model initial margins with $\GG$-adapted processes $I^{TC} = (I^{TC}_t)_{t \,\in\, [0,T]},$ $I^{FC} = (I^{FC}_t)_{t \,\in\, 
[0,T]},$ and we denote by $B^{I,x}, \, x \,\in\, \{l,b\},$ the cash accounts associated to $I^{TC}, \, I^{FC},$ respectively.   

In case the bank is the \emph{initial margin provider,} the bank posts $I^{TC}$ to the counterparty (TC), i.e.
		\begin{equation}\label{IM_TC} 
			\psi^{I,l}_tB^{I,l}_t = -I^{TC}_t, \,t \,\in\, [0,T], 
			\end{equation}
			or equivalently
			\begin{equation}\label{IM_TC2}
			-\psi^{I,l}_tdB^{I,l}_t = r^{I,l}_tI^{TC}_tdt \;. 
			\end{equation}
In case the bank is the \emph{initial margin taker,} the bank receives $I^{FC}$ from the counterparty (FC), i.e. 
			\begin{equation}\label{IM_FC} 
			\psi^{I,b}_tB^{I,b}_t = I^{FC}_t, \,t \,\in\, [0,T],    
			\end{equation}
			or equivalently  
			\begin{equation}\label{IM_FC2}
			\psi^{I,b}_tdB^{I,b}_t= r^{I,b}_tI^{FC}_tdt \;. 
			\end{equation}

We highlight that, contrary to the case of variation margin and in line with \textcolor{black}{current market practice}, 
$I^{TC}$ and $I^{FC}$ are simultaneously active and do not net each other. 
In the sequel we use  $I = (I_t)_{t\,\in\,[0,T]}$ as shorthand for both presented or received initial margins.

\textcolor{black}{We now introduce the definition of self-financing strategy in our market model. We recall that we assume the point of 
view of the bank (i.e. the hedger)}\textcolor{black}{.} 
\begin{definition} \label{def:dynamic_ptf} 
A \emph{dynamic portfolio,} denoted by $\varphi,$ is given by
\begin{align*}
\varphi = \left(\xi^1,\ldots,\xi^d,\xi^B,\xi^C,\psi^{1},\ldots,\psi^{d},\psi^B,\psi^C,\psi^{f,b},\psi^{f,l},\psi^{c,b},
					\psi^{c,l},\psi^{I,b},\psi^{I,l} \right), 
\end{align*}
where 
\begin{itemize}
\item[(i)] $\xi^1,\ldots,\xi^d$ are $\GG$-predictable processes, denoting the number of shares of the risky primary assets $S^1,\ldots,S^d$.

\item[(ii)] $\xi^B,\xi^C$ are $\GG$-predictable processes, denoting the number of shares of the risky bonds $P^B$ and $P^C$.

\item[(iii)] $\psi^{1},\ldots,\psi^{d},\psi^B,\psi^C$ are $\GG$-adapted processes, denoting the number of shares of the repo accounts 
			$B^{1},\ldots B^{d},B^B,B^C$.
			
\item[(iv)] $\psi^{f,b}$ is a $\GG$-adapted process, denoting the number of shares of the unsecured funding borrowing cash account $\Bfk{b}$. 

\item[(v)] $\psi^{f,l}$ is a $\GG$-adapted process, denoting the number of shares of the unsecured funding lending cash account $\Bfk{l}$. 

\item[(vi)] $\psi^{c,b}$ is a $\GG$-adapted process, denoting the number of shares of the collateral borrowing cash account $\Bck{b}$ for 
			the received cash collateral. 
			
\item[(vii)] $\psi^{c,l}$ is a $\GG$-adapted process, denoting the number of shares of the collateral lending cash account $\Bck{l}$ for the 
			posted cash collateral. 
			
\item[(viii)] $\psi^{I,b}$ is a $\GG$-adapted process, denoting the number of shares of the initial margin borrowing cash account $\BIk{b}$ for 
			the initial margin received from the counterparty.
			
\item[(ix)] $\psi^{I,l}$ is a $\GG$-adapted process, denoting the number of shares of the initial margin lending cash account $\BIk{l}$ for 
			the initial margin posted to the counterparty.
\end{itemize}
All processes introduced above are such that the stochastic integrals in the sequel are well defined.
\end{definition}

Given a dynamic portfolio, we associate it to a financial contract, known in the literature as \emph{Credit Support Annex (CSA),} 
see e.g. \cite{BCBS2014}.  

\begin{definition} 
A CSA between the bank and the counterparty is represented \textcolor{black}{by} the pair $(C,I),$ where $C$ is the variation margin and 
$I$ is the initial margin.  
\end{definition}

\begin{definition} \label{def:coll_hedger_strategy}
A \emph{collateralized hedger's trading strategy} associated to the collateralized contract $\bar{A}$ and the CSA $(C,I)$ is a 
quintuplet $\left(x,\varphi,\bar{A},C,I\right),$ where $x \,\in\,\RR$ is the initial endowment, $\varphi$ is a dynamic portfolio 
\textcolor{black}{as in Definition \ref{def:dynamic_ptf}} and $I$ is the initial margin.
\end{definition}

\begin{definition}
Given the initial endowment $x,$ a collateralized hedger's trading strategy $\left(x,\varphi,\bar{A},C,I\right)$ associated 
to the collateralized contract $\bar{A}$ and the CSA $(C,I)$ is said to be \emph{self-financing} if, for any $t \,\in\, [0,T],$ 
the \emph{wealth process} $V_t(\varphi)$ given by 
\begin{align}\label{def:wealth_prox}
\begin{aligned}
V_t(\varphi) & := \psi^{f,b}_tB^{f,b}_t + \psi^{f,l}_tB^{f,l}_t 
- \left(\psi^{c,b}_tB^{c,b}_t+\psi^{c,l}_tB^{c,l}_t+\psi^{I,l}_t\BIk{l}_t \right) \;, 
\end{aligned}
\end{align}
satisfies 
\begin{align}
\label{eq:sefFinancing}
\begin{aligned}
V_t(\varphi) & = x + \sum_{i=1}^d\int_{(0,t]}\xi^i_u\left(\mu^i(u,S_u)du + \sum_{k=1}^d\sigma^{i,k}(u,S_u)dW^{k,\PP}_u 
+ \kappa^i(u,S_u)du\right)\\
& + \sum_{i=1}^{d}\int_0^t\psi^i_u dB^i_u + \sum_{j\in\{B,C\}}\int_0^t\left(\xi^j_udP^j_u + \psi^j_udB^j_u\right) 
- \bar{A}_t\\
& + \int_0^t\psi^{f,b}_u dB^{f,b}_u + \int_0^t\psi^{f,l}_udB^{f,l}_u - \int_0^t\psi^{c,b}_udB^{c,b}_u 
- \int_0^t\psi^{c,l}_udB^{c,l}_u - \int_0^t\psi^{I,b}_u dB^{I,b}_u-\int_0^t\psi^{I,l}_udB^{I,l}_u. 
\end{aligned}
\end{align}
\end{definition}

The last two terms in \eqref{eq:sefFinancing} represent the cash for the received initial margin. 
In general, we assume zero initial endowment, $x=0$, i.e., $V_t(\varphi)=V_t\left(0,\varphi,\bar{A},C,I\right)$ 
for the sake of simplicity. \textcolor{black}{Note that the sign minus in front of the last term in \eqref{def:wealth_prox} 
depends on our convention on the collateral. Moreover, we are not including the cash account for the received initial margin. 
This is due to the fact that the received initial margin is posted in a segregated account and, hence, is not available 
as a funding asset to the trading desk.}

\begin{definition}\label{def:admissibility}
A collateralized hedger's trading strategy is \emph{admissible} if it is self-financing and the associated value process $V(\varphi)$ is bounded \textcolor{black}{from below}.  
\end{definition} 

We provide the following definition of \emph{arbitrage-free strategy.} 
\begin{definition} \label{def:ArbitrageFree}
\textcolor{black}{The} market \textcolor{black}{model} is \emph{arbitrage-free} if for
\textcolor{black}{any admissible hedger's trading strategy of the form}
$\left(0,\varphi,0,0,0 \right)$ we have either 
$$\PP\left[V_{\textcolor{black}{\tau}}\left(0,\varphi,0,0,0 \right) = 0 \right] = 1 \,\mbox{ or } \, 
\PP\left[V_{\textcolor{black}{\tau}}\left(0,\varphi,0,0,0 \right) < 0 \right] > 0 \textcolor{black}{.}$$ 
\end{definition}

\begin{remark}
In \cite[Definition 3.3]{BieRut15} the authors introduce the concept of a market which is said to be \textit{arbitrage-free for the hedger} 
with respect to a class of contingent claims. Their definition is formulated in terms of a \textit{netted wealth process}, which 
corresponds to a long-short strategy involving the claim $\bar{A}$, where the first position is hedged and the second is unhedged.
On the other hand, in \cite{BiCaStu2018} the question concerning absence of arbitrage is first answered in a setting where only the 
basic traded assets are considered. This is also referred to as \emph{absence of arbitrage with respect to the null contract} in 
\cite{biciarut2018}. In our setting, the two approaches coincide. 
\end{remark}

\begin{assumption}\label{assumptionRates}
We assume $r^{f,l}_t\leq r^{f,b}_t$, $d\PP\otimes dt$-a.s.
\end{assumption}

\textcolor{black}{We now} prove absence of arbitrage for \textcolor{black}{the market model in our setting.} 

\begin{proposition}\label{pro:noArbitrageBasicMarket}
Let Assumption~\ref{assumptionRates} hold. Moreover, assume there exists a probability 
measure $\QQ\sim \PP$ such that the discounted asset price processes
\begin{align} \label{eq:SPQ}
\tilde{S}^{i,cld}_t & := \frac{S^{i,cld}_t}{B^{i}_t}, \ i=1,\ldots d, \quad\quad 
\tilde{P}^j_t := \frac{P^{j}_t}{B^{j}_t},  \ j\in\{B,C\}, 
\end{align}
are local martingales. 
Then, the market \textcolor{black}{model} is free of arbitrage opportunities.
\end{proposition}

\begin{proof} 
See Section \ref{app:A1} in Appendix \ref{app}. 
\end{proof}

From now on, we assume the following. 
\begin{assumption} \label{ass:true_martingale}
There exists an equivalent martingale probability measure $\QQ\sim\PP$ under which 
the processes $\tilde{S}^{i,cld}_t, \, \tilde{P}^j_t$ in \eqref{eq:SPQ} are local martingales with dynamics 
\begin{align*}
d\tilde{S}^{i,cld}_t & = \frac{1}{B^{i}_t}\left(dS^i_t-r^i_tS^i_tdt+dD^i_t\right) 
= \sum_{k=1}^d\frac{\sigma^{i,k}(t,S_t)}{B^{i}_t}dW^{k,\QQ}_t, \ i=1,\ldots, d,\\
d\tilde{P}^j_t & = \frac{1}{B^{j}_t}\left(dP^j_t-r^j_tP^j_{t-}dt\right) = -\tilde{P}^j_{t-}dM^{j,\QQ}_t, \ j\in\{B,C\}, 
\end{align*}
\textcolor{black}{
where 
$$M^{j,\QQ}_t := H^{j,\QQ}_t - \int_0^{t\wedge\tau^j} \lambda^{j,\QQ}_s ds,$$ 
with 
$$\lambda^{j,\QQ}_t := r^f_t - \lambda^{j,\PP}_t - r^j_t, \; t \,\in\, [0,T], \, j\,\in\, \{B,C\}.$$ 	
}
\end{assumption} 

\textcolor{black}{
By \cite{AkFon2019} we obtain that the predictable representation property in $\GG$ still holds after the change of measure to $\QQ$ with respect to $W^{\QQ}, \, M^{j,\QQ}, \, j \,\in\, \{B,C\}.$} 
\textcolor{black}{
\begin{definition} \label{def:spaces} 
	Let	$\beta \geq 0.$ 
	The subspace of all $\RR^d$-valued, $\FF$-adapted processes $X$ such that  
	\begin{align} \label{eq:spaceH2}
		\Ex{\QQ}{\int_0^T e^{\beta t} \left|X_t\right|^2 dt}<\infty 
	\end{align}
	is denoted by $\cH^{2,d}_{\beta,T}(\QQ).$ We set $\cH^{2\textcolor{black}{,d}}(\QQ) := \cH^{2,\textcolor{black}{,d}}_{0,T}(\QQ)$ \textcolor{black}{and}
	\begin{equation} \label{def:NormH}
		\left\| X \right\|_{\cH^{2,d}_{\beta,T}} \textcolor{black}{:} = \sqrt{\Ex{\QQ}{\int_0^T e^{\beta t} \left|X_t\right|^2 dt}}. 
	\end{equation}	
	The subspace of all $\RR^d$-valued, $\FF$-adapted processes $X$ such that 
	\begin{align} \label{eq:spaceS2b}
		\EE^{\QQ}\left[\sup_{t \,\in\, [0,T]} e^{\beta t} \left|X_t\right|^2 \right] < \infty 
	\end{align}
	is denoted by $\cS^{2,d}_{\beta,T}(\QQ).$ We set $\cS^{2}(\QQ) := \cS^{2,1}_{\textcolor{black}{0},T}(\QQ)$ \textcolor{black}{and}
	\begin{equation} \label{def:NormS}
		\left\| X \right\|_{\cS^{2,d}_{\beta,T}} \textcolor{black}{:}= \sqrt{\EE^{\QQ}\left[\sup_{t \,\in\, [0,T]} e^{\beta t} \left|X_t\right|^2 \right]}. 
	\end{equation}
We denote by $\cH^{2,2}_{\lambda}(\QQ)$ the space of $\FF$-adapted processes $X$ with values in $\mathbb{R}^2$ such that 
\begin{align} \label{eq:spaceHN}
	\EE^{\QQ}\left[\int_0^T\left|X^1_t\right|^2 \lambda^{B,\QQ}_t+\left|X^2_t\right|^2 \lambda^{C,\QQ}_t dt \right] < \infty 
\end{align}
with 
\begin{equation} \label{def:NormS_HN}
	\left\| X \right\|_{\cH^{2,2}_{\lambda}} = \sqrt{\EE^{\QQ}\left[\int_0^T\left|X^1_t\right|^2 \lambda^{B,\QQ}_t+\left|X^2_t\right|^2 \lambda^{C,\QQ}_t dt \right]}. 
\end{equation}
\end{definition} 
}

\textcolor{black}{
\begin{assumption} \label{ass:proxA} 
	Assume that \textcolor{black}{the payment stream} $A\in \cS^2(\QQ).$ 
\end{assumption}
}

\section{xVA framework}\label{Sect:single} 
\textcolor{black}{ 
The main contribution of this Section is manifold. On the one hand, in an arbitrage-free framework we determine the unique solution to the $\mathbb{G}$-BSDE for the valuation of the hedger portfolio $(x, \varphi,\bar{A},C,I),$ consisting of risk-less and risky assets, defaultable bonds and contingent claims. 
On the other hand, we provide an arbitrage-free framework for clean values  and the market practice of using contingent claim specific discounting regimes (CSA discounting) by the introduction of a new value adjustment. 
Finally, we investigate the shape of the $\GG$-BSDE in the realistic situation where the portfolio has a hierarchical structure of sub-portfolios, and consider non-linear effects in portfolio's valuation.}

Under Assumption \ref{ass:true_martingale} the dynamics of a self-financing collateralized trading strategy 
$\left(x,\varphi,\bar{A},C,I\right)$ is 
\begin{align} \nonumber
dV_t(\varphi)&=\sum_{i=1}^d\xi^i_tB^i_td\tilde{S}^{i,cld}_t+\sum_{j\in\{B,C\}}\xi^j_tB^j_td\tilde{P}^j_t-d\bar{A}_t\\ \label{eq:coll_trad}
&\quad+\psi^{f,l}_tdB^{f,l}_t+\psi^{f,b}_tdB^{f,b}_t-\psi^{c,l}_tdB^{c,l}_t-\psi^{c,b}_tdB^{c,b}_t-\psi^{I,l}_tdB^{I,l}_t 
- \psi^{I,b}_tdB^{I,b}_t \;.  
\end{align}
By using the repo constraints \eqref{eq:repoConstraint}, \eqref{eq:coll_bank}, \eqref{eq:coll_counterparty} and \eqref{IM_TC2}, the portfolio value satisfies 
$$V_t(\varphi)=\psi^{f,l}_tB^{f,l}_t+\psi^{f,b}_tB^{f,b}_t+C_t+I^{TC}_t, \, t\,\in\,[0,T],$$ 
since $I^{FC}$ is segregated. 

Thanks to \eqref{eq:bl_sim} we obtain the identities 
\begin{align} \label{eq:psiFL}
\psi^{f,l}_t & = \left(V_t(\varphi)-C_t-I^{TC}_t\right)^+\left(B^{f,l}_t\right)^{-1}, \\ \label{eq:psiFB}
\psi^{f,b}_t & = -\left(V_t(\varphi)-C_t-I^{TC}_t\right)^-\left(B^{f,b}_t\right)^{-1}
\end{align}
for $t\,\in\,[0,T].$ 
Observe that by \eqref{eq:coll_bank} and \eqref{eq:coll_counterparty}
\begin{align} \label{Coll:condb2}
-\psi^{c,l}_tdB^{c,l}_t & = -\psi^{c,l}_t r^{c,l}_t B^{c,l}_t dt = +r^{c,l}_tC^+_tdt , \\ \label{Coll:condl2}
-\psi^{c,b}_tdB^{c,b}_t & = -\psi^{c,b}_t r^{c,b}_t B^{c,b}_t dt = -r^{c,b}_tC^-_tdt ,
\end{align}
respectively. 
By \eqref{eq:psiFL}, \eqref{eq:psiFB}, \eqref{Coll:condb2}, \eqref{Coll:condl2}, \eqref{IM_TC2} and \eqref{IM_FC2}, we can rewrite 
the wealth dynamics as follows
\begin{align}
\label{eq:dynamics} 
\begin{aligned}
dV_t(\varphi)& = \sum_{i=1}^d\xi^i_tB^i_td\tilde{S}^{i,cld}_t + \sum_{j\in\{B,C\}}\xi^j_tB^j_td\tilde{P}^j_t-d\bar{A}_t \\  
 & +\left[r^{f,l}_t\left(V_t(\varphi)-C_t-I^{TC}_t\right)^+ - r^{f,b}_t\left(V_t(\varphi)-C_t-I^{TC}_t\right)^-\right.\\ 
 & \left.+ r^{c,l}_tC^+_t - r^{c,b}_tC^-_t + r^{I,l}_tI^{TC}_t - r^{I,b}_tI^{FC}_t\right]dt.
\end{aligned}
\end{align}

We now introduce for convenience an auxiliary artificial interest rate process $r=\left(r_t\right)_{t\,\in\,[0,T]}$,
 assumed to be right-continuous, bounded and $\FF$-adapted. This rate is not necessarily linked to a traded asset, 
but it can be interpreted as an interest rate level, used to express all other rates as spreads over this artificial rate. 
When needed, we will explicitly state when the rate $r$ becomes a market rate. Using the artificial rate $r,$ we can conveniently 
rewrite the portfolio dynamics as follows 
\begin{align} \nonumber 
dV_t(\varphi) & = \sum_{i=1}^d\xi^i_tB^i_td\tilde{S}^{i,cld}_t+\sum_{j\in\{B,C\}}\xi^j_tB^j_td\tilde{P}^j_t-d\bar{A}_t\\ \nonumber 
& +\left[(r^{f,l}_t-r_t)\left(V_t(\varphi)-C_t-I^{TC}_t\right)^+-(r^{f,b}_t-r_t)\left(V_t(\varphi) - 
C_t-I^{TC}_t\right)^-\right.\\ \nonumber 
& \left.+(r^{c,l}_t-r_t)C^+_t-(r^{c,b}_t-r_t)C^-_t+(r^{I,l}_t-r_t)I^{TC}_t-r^{I,b}_tI^{FC}_t+r_tV_t(\varphi)\right]dt \\ \label{dynamics2}
& = \textcolor{black}{\sum_{i,k=1}^d \xi^i_t \sigma^{i,k}(t,S_t) d W_t^{k,\mathbb{Q}} - \sum_{j\in\{B,C\}}\xi^j_tB^j_t 
 \tilde{P}^j_{t-} d M_t^{j,\mathbb{Q}}} - d\bar{A}_t\\  \nonumber 
& +\left[(r^{f,l}_t-r_t)\left(V_t(\varphi)-C_t-I^{TC}_t\right)^+-(r^{f,b}_t-r_t)\left(V_t(\varphi) - 
C_t-I^{TC}_t\right)^-\right.\\ \nonumber 
& \left.+(r^{c,l}_t-r_t)C^+_t-(r^{c,b}_t-r_t)C^-_t+(r^{I,l}_t-r_t)I^{TC}_t-r^{I,b}_tI^{FC}_t+r_tV_t(\varphi)\right]dt 
\;,  
\end{align}
where we added and subtracted the term $r_tV_t(\varphi) dt$ \textcolor{black}{and used Assumption \ref{ass:true_martingale}.} 
\begin{remark}The term $(r^{I,l}_t-r_t)I^{TC}_t$ measures a funding benefit from the posted initial margin over the reference rate 
level $r$. We would like to stress that, in general, spreads over $r$ can be negative, representing that we may have funding costs, 
even when the bank is collateral provider. Such a situation is faced by banks, which clear swaps with the \emph{London Clearing House} 
(LCH). If $r$ is chosen to represent the EONIA overnight rate, then the rate applied by LCH is $r^{I,l}=r-58bps,$ where bps stands for
basis points\footnote{see https://www.lch.com/risk-collateral-management/ltd-collateral-management/ltd-fees-collateral}. 
On top of such a negative benefit, the bank needs to take into account the cost of raising the amount $I^{TC}$, hence initial margin can generate funding costs in both directions, from the point of view of fund-raising and from the point of view of collateral remuneration, hence representing a significant source of costs for the bank.
\end{remark}

\textcolor{black}{ 
In case of default, cashflows are exchanged between the surviving agent and the liquidators of the defaulted agent.
Here we use the term \emph{agent} as a placeholder for the bank or for the counterparty. Due to the exchange of 
cashflows at default time, agents need to perform a valuation of the position at a random time 
\textcolor{black}{represented by the close-out condition, see \cite{BiCaStu2018} Section 3.4}. 
The object of the analysis can be the value in the absence of counterparty risk (referred to in the literature as 
\emph{risk-free close-out}) or the value of the trade including the price adjustments due to counterparty risk and 
funding (\emph{risky close-out}), see e.g. \cite{BriMor2018}. 
A risky close-out condition guarantees that the surviving counterparty can ideally 
fully substitute the transaction with a new trade entered with another counterparty with the same credit quality. 
This comes at the price of a significant increase of the complexity of the valuation equations. Current 
market practice and the existing literature mainly focus on the case of risk-free close-out value. 
\begin{definition}Let $0<R^j<1, \; j\,\in\,\{B,C\},$ be the recovery rates
of the bank and the counterparty, respectively.
The close-out condition $\theta_\tau(\cV,C,I^{TC},I^{FC}),$ expressed from the bank's perspective, is defined by
\begin{equation} \label{eq:finalCloseOut} 
\begin{aligned}
	\theta_\tau(\cV,C,I^{TC},I^{FC}) & := \cV_\tau + \Ind{\tau^C< \tau^B}(1-R^C)\left(\cV_\tau - C_{\tau-} + I^{FC}_{\tau-}\right)^-\\
	&\quad-\Ind{\tau^B< \tau^C}(1-R^B)\left(\cV_\tau - C_{\tau-} - I^{TC}_{\tau-}\right)^+. 
\end{aligned}
\end{equation} 
\end{definition} 
}

We restate the portfolio dynamics \textcolor{black}{under $\mathbb{Q}$} in the form of a BSDE.  
We set 
\begin{subequations}
\begin{align} 
\label{eq:Zk}
Z^k_t & :=\frac{\sum_{i=1}^d\xi^i_t\sigma^{i,k}(t,S_t)}{B_t^i}, \\ \label{eq:Uj}
U^j_t & :=-\xi^j_t \tilde{P}^j_{t-},\\ \label{eq:driver}
f(t,V,C,I^{TC},I^{FC}) & :=-\left[(r^{f,l}_t-r_t)\left(V_t(\varphi)-C_t-I^{TC}_t\right)^+-(r^{f,b}_t-r_t)\left(V_t(\varphi)-C_t-I^{TC}_t\right)^-
\right.\\ \nonumber
&\left.+(r^{c,l}_t-r_t)C^+_t-(r^{c,b}_t-r_t)C^-_t+(r^{I,l}_t-r_t)I^{TC}_t-r^{I,b}_tI^{FC}_t\right]. 
\end{align}
\end{subequations}
\textcolor{black}{Under the close-out condition \eqref{eq:finalCloseOut},}
the $\GG$-BSDE for the portfolio's dynamics \textcolor{black}{in \eqref{dynamics2} under $\mathbb{Q}$} has then the form 
\begin{align} \label{eq:GBSDEdiff}
\begin{cases}
-dV_t(\varphi) = d\bar{A}_t + \left(f(t,V,C,I^{TC},I^{FC}) - r_tV_t(\varphi)\right)dt 
-\sum_{k=1}^dZ^k_tdW^{k,\QQ}_t - \sum_{j\in\{B,C\}}U^j_tdM^{j,\QQ}_t\\
V_\tau(\varphi) = \theta_\tau(\cV,C,I^{TC},I^{FC})\textcolor{black}{,}
\end{cases}
\end{align}
\textcolor{black}{for $t \leq \tau.$}
We prove in Theorem \ref{thm:solGbsde} that there exists a unique solution $(V,Z,U)$ for the $\GG$-BSDE \eqref{eq:GBSDEdiff}, 
and the process $V$ assumes the following form on $\{\tau > t\}$
\begin{align}\label{eq:GBSDEint} 
   V_t(\varphi) = B^r_t\Excond{\QQ}{\int_{(t,\tau \wedge T]}\frac{d\bar{A}_u}{B^r_u} 
		+ \int_t^{\tau \wedge T}\frac{f(u,V,C,I^{TC},I^{FC})}{B^r_u}du 
		+ \Ind{\tau\leq T}\frac{\theta_\tau(\cV,C,I^{TC},I^{FC})}{B^r_\tau}}{\cG_t}, 
\end{align} 
where $B_t^r := \exp\left(\int_0^t r_u du \right), \, t\,\in\, [0,T].$ 
{\textcolor{black}{
\begin{remark}
The $\GG$-BSDE \eqref{eq:GBSDEdiff} is in line with the current market practice on xVA. Regulatory changes in the bankruptcy procedures can be encoded via changes in the terminal condition (i.e. the close-out condition) whereas new funding policies can be captured by suitable adjustments of the driver.
\end{remark}
In order to solve \eqref{eq:GBSDEdiff} we first have to specify our choice for $\mathcal{V}$ in the close-out condition \textcolor{black}{\eqref{eq:finalCloseOut}}. To this purpose, we need the results of the following Section.} 

\subsection{Clean Value under $\FF$} \label{sect:clean}
A financial product can be traded between any two counterparties. Since every agent has a different credit quality and different 
funding costs, this means in general that a single product (e.g. a 10 year EUR swap) has as many potential values as the number of 
possible combinations of agents in the market. It would be highly impractical for a broker to publish all possible market quotes for all possible counterparties. 
\textcolor{black}{
Publicly observable market quotes provided by data providers do not take into account xVA frictions and are typically given by a single value (more precisely a bid and 
offer price) called \emph{clean price}. We now provide an arbitrage-free valuation for clean prices consistent with the current market practice for market quotes\footnote{Some authors criticize  the idea of a price decomposition in terms of clean value and an adjustment process, see e.g. \cite{biciarut2018}. A situation where the decomposition is not justified is given e.g. by a contingent claim whose dividend process depends in a non-linear way on the strategy of the hedger in a fully non-linear market, where we have bid offer spread in all rates, including the repo rates of the assets (which are zero in our case). However, payoffs traded on the market do not usually  feature such non-linear effects. Furthermore, publicly available market quotes  are not influenced by counterparty credit risk.}. 
}

A \emph{clean price} is an ideal value process that would be acceptable between two agents entering a perfectly collateralized 
transaction. Perfect collateralization however is not enough to produce a clean price: we also need to explicitly assume that 
the two agents entering the transaction are default-free. This is necessary because, even in the presence of a perfect ideal 
collateral agreement, counterparty risk is not perfectly annihilated: when a counterparty defaults, she stops posting collateral. 

However, default is not automatically legally recognized: typically, bankruptcy procedures require some days (e.g. 10 or 
20 days) before the close-out payments are exchanged. This creates a period of time where the counterparty is not officially 
defaulted but without any collateral adjustment. Such period of time is known as \emph{margin period of risk.} 
During such interval of time the value of the claim deviates from the value of the collateral account thus creating a 
credit exposure. 

Hence, to preclude margin period of risk and obtain the ideal clean price process, we consider a parallel market model with perfect collateralization but no default risk. 

\begin{assumption}[Clean market]\label{cleanMarket} 
A clean market under $\FF$ without bid-offer spreads is defined by 
\begin{itemize}
\item[(i)] no bid-offer spread in the funding accounts, i.e., $r^{f,l}_t=r^{f,b}_t=r^f$; 
\item[(ii)] no bid-offer spread in the collateral accounts, i.e.,  $r^{c,l}_t=r^{c,b}_t=r^c$; 
\item[(iii)] the collateral rate is equal to the fictious rate, i.e., $r^c=r$; 
\item[(iv)] there is no default, i.e. $\tau^B=\tau^C=\infty$ 
and risky bonds are excluded from the market; 
\item[(v)] there is no exchange of initial margin; 
\item[(vi)] perfect collateralization, i.e., $\hat{V}_t\equiv C_t$, for all $t \,\in\, [0,T],$ where we use $\hat{V}$ 
			to denote the value process of a collateralized hedging strategy in the fictious market without default-risk.
\end{itemize}
\end{assumption}

Note that \emph{(vi)} in Assumption \ref{cleanMarket} implies that the portfolio weights in the cash accounts are of 
the form
\begin{align*}
\psi^c_t=-\frac{\hat{V}_t}{B^c_t}, \quad \psi^f_t\equiv 0,  \,\mbox{ for all }\, t\,\in\,[0,T], 
\end{align*}
meaning that the position is totally funded by the collateralization scheme, and $\hat{V} = (\hat{V}_t)_{[0,T]}$ 
is an $\FF$-adapted process.

\textcolor{black}{The dynamics of the clean portfolio value $\hat{V}$ resulting from \eqref{dynamics2} and Assumption~\ref{cleanMarket} 
are given by}

\begin{align}\label{cleanPortfolio}
\begin{aligned}
d\hat{V}_t(\varphi) & = \sum_{k=1}^d\hat{Z}^k_tdW^{k,\QQ}_t-dA_t+r_t\hat{V}_t(\varphi)dt, \,\mbox{ where } \, 
\hat{Z}^k_t :=\sum_{i=1}^d\hat{\xi}^i_t\sigma^{i,k}(t,S_t). 
\end{aligned}
\end{align}

Inserting the terminal condition $\hat{V}_T=0,$ we can rewrite \textcolor{black}{the $\mathbb{F}$-dynamics \eqref{cleanPortfolio} for $\hat{V}$ under $\mathbb{\QQ}$} in the classical $\FF$-BSDE form 

\begin{align} \label{cleanBSDE}
\begin{cases}
-d\hat{V}_t(\varphi) = dA_t-r_t\hat{V}_t(\varphi)dt -\sum_{k=1}^d\hat{Z}^k_tdW^{k,\QQ}_t\\
\hat{V}_T(\varphi) = 0.
\end{cases}
\end{align}

We now prove existence and uniqueness for the solution of \eqref{cleanBSDE}. 

\begin{theorem}\label{thm:solFbsde}
Under Assumption~\ref{ass:proxA} on $A,$ 
there exists a unique solution $\left(\hat{V},\hat{Z} \right) \,\in\, \cS^{2}(\QQ) \times \cH^{2,d}(\QQ)$
to the clean \textcolor{black}{$\mathbb{F}$}-BSDE \eqref{cleanBSDE}.
\end{theorem}

\begin{proof}
We note that the clean BSDE \eqref{cleanBSDE} is similar to the linear BSDE studied e.g. in  \cite{elpq1997}, 
where the driver is the multidimensional Brownian motion $\left(W^{1,\QQ},\ldots,W^{d,\QQ}\right)^\top$. 

We can apply \cite[Theorem 4.1]{nr2016} by observing that $M = W^\QQ$, $Q_t = t$, $U = A$, 
$\hat{V} = Y$ and $h(t,Y_t,Z_t) = -r_t\hat{V}_t$, which clearly fulfills the uniform Lipschitz condition. 
Also the condition $h(\cdot,0,0)\in\cS^2(\QQ)$ is trivially satisfied. 
We also observe that $\XX = S = diag(S^1,\ldots,S^d)$, hence we have $m_t = \sigma(t,S_t),$ so that 
$\gamma_t = S^{-1}_1\sigma(t,S_t),$ for $\gamma$ satisfying the ellipticity condition \eqref{eq:ellipticity}.
According to Theorem~\ref{th:existenceUniquenessNR} we have $\hat{V} \,\in\, \cH^{2}(\QQ)$ and $\hat{V}-A \,\in\, \cS^{2}(\QQ)$. 
Now, Assumption~\ref{ass:proxA} allows us to conclude that also $\hat{V}\in\cS^{2}\textcolor{black}{(\mathbb{Q})}$.
\end{proof}

Next we show that the process $\hat{V}$ in Theorem~\ref{thm:solFbsde} provides the arbitrage-free price for the contract 
with cashflow stream $A$.

\begin{theorem}\label{th:cleanPriceTh}
Let $\QQ\sim\PP$ be an equivalent  probability measure such that all processes $\tilde{S}^{i,cld}$, $i=1,\ldots,d,$ are local 
$\QQ$-martingales. Let $\left(\hat{V},\hat{Z} \right)$ be the unique solution of 
\eqref{cleanBSDE}. 
Then, under Assumption \ref{ass:proxA} on $A,$ we have 
\begin{align}\label{fair_price:clean_mkt}
\hat{V}_t(\varphi):=\Excond{\QQ}{B^r_t\int_{(t,T]} \frac{dA_u}{B^r_u}}{\cF_t}, \,\mbox{for all } t \,\in\, [0,T] \;.
\end{align}
\end{theorem}

\begin{proof}
The proof follows because $\hat{Z} \,\in\, \cH^{2,d}(\QQ)$ by Theorem \ref{thm:solFbsde}, $B$ is bounded and thanks to 
Assumption \ref{ass:proxA}.  
\end{proof}

\textcolor{black}{
From now on, we assume to work with the \cadlag version of $\hat{V}.$ }
\begin{remark}  
Here we introduce the concept of clean value by means of a replicating strategy in a fictious idealized market. Our constructive approach is in line with the \textcolor{black}{current} market standard \textcolor{black}{and the concept of \emph{third-party valuation} of \mbox{\cite{BiCaStu2018}}}. 
Formula \eqref{fair_price:clean_mkt} encodes the idea of \emph{CSA discounting.} 
Since the rate $r$ is the remuneration of collateral in a stylized perfect collateral agreement, we do not need to 
postulate the existence of a risk-free rate. \cite{BiCaStu2018} define the clean value by introducing an additional 
valuation measure different from $\QQ.$ 
Working with the pricing measure $\QQ$ also avoids the issue of estimating parameters under different measures. 
\end{remark}

So far, our discussion of the clean market focused on a dividend process specified under the reference filtration $\FF$. 
As stressed e.g. in \cite{crepey2015b}, this assumption is too restrictive to e.g. cover credit derivatives or \emph{wrong-way risk.} 
Though, our objective is to focus on multiple aggregation levels and different discounting regimes, hence we choose to avoid the technicalities that are involved in generalizations of the immersion hypothesis.

\subsection{\textcolor{black}{Portfolio valuation via $\GG$-BSDE}}
\textcolor{black}{
We now discuss existence and uniqueness of the solution for the \textcolor{black}{$\mathbb{G}$}-BSDE \eqref{eq:GBSDEdiff} by following the approach of \cite{crepey2015b}. 
To this purpose we use the results on the clean value of Section \ref{sect:clean} and the following Assumption.}

\begin{assumption}\label{assumptionCloseOut}
We assume a risk-free close-out valuation under $\FF$, namely we set $\cV_t=\hat{V}_t(\varphi)$ in \textcolor{black}{the close-out condition}  \eqref{eq:finalCloseOut}. 
\end{assumption}

\begin{definition} \label{def:XVA}
We define the following valuation adjustments:
\begin{align*}
CVA_t& := B^r_t\Excond{\QQ}{
	\Ind{\tau^C< \tau^B}(1-R^C)\frac{1}{B^r_\tau}\left(\hat{V}_\tau(\varphi)-C_{\tau-} + I^{FC}_{\tau-}\right)^-}{\cG_t}, \\
DVA_t& := B^r_t\Excond{\QQ}{
	\Ind{\tau^B< \tau^C}(1-R^B)\frac{1}{B^r_\tau}\left(\hat{V}_\tau(\varphi)-C_{\tau-} - I^{TC}_{\tau-}\right)^+}{\cG_t}, \\
FVA_t& :=B^r_t \Excond{\QQ}{\int_t^{\textcolor{black}{\tau}} 
	\frac{(r^{f,l}_u-r_u)\left(V_u(\varphi)-C_u-I^{TC}_u\right)^+ - (r^{f,b}_u-r_u)\left(V_u(\varphi)-C_u-I^{TC}_u\right)^-}{B^r_u}du}{\cG_t}, \\
ColVA_t& :=B^r_t \Excond{\QQ}{\int_t^{\textcolor{black}{\tau}} 
	 \frac{(r^{c,l}_u-r_u)C^+_u-(r^{c,b}_u-r_u)C^-_u}{B^r_u}du}{\cG_t}, \\
MVA_t& := B^r_t\Excond{\QQ}{\int_t^{\textcolor{black}{\tau}} 
	\frac{(r^{I,l}_u-r_u)I^{TC}_u-r^{I,b}_uI^{FC}_u}{B^r_u}du}{\cG_t}.
\end{align*}
On $\left\{\tau > t \right\},$ we define 
\begin{align}\label{eq:XVA}
XVA_t := -CVA_t + DVA_t + FVA_t + ColVA_t + MVA_t,  
\end{align}
and set 
\begin{equation} \label{eq:XVA_tau}
	XVA_{\tau} = -\theta_{\tau} + \hat{V}_{\tau} \,\mbox{ on } \left\{\tau \leq t \right\},
\end{equation} 
where $\theta_{\tau}$ is defined in \eqref{eq:finalCloseOut}. 
\end{definition}

\textcolor{black}{Note that} the FVA term in Definition~\ref{def:XVA} has a recursive nature, \textcolor{black}{see also \mbox{\cite{pit10}}}. 
The exposure is proportional to the full value of the transaction $V$ and not only to the clean value $\hat{V}$. 
This implies a high complexity of the numerical scheme. 
Some practitioner's papers, such as \mbox{\cite{bj2013}}, avoid the recursivity issue by means of ad-hoc choices of the funding 
strategies, such as the funding strategy called \emph{semi-replication with no shortfall on default.} 
However, the bank usually needs to fund the clean value and the value adjustments. 
Hence, this feature cannot be ignored in a comprehensive mathematical model. 

\begin{example}
Set $I^{TC}_t=I^{FC}_t=0$, $r^{f,b}=r^{f,l}=r^f$, $r^{c,b}=r^{c,l}=r^c$ and $\tau^C=\tau^B=\infty$. 
Then the driver of the full BSDE is given by
\begin{equation} \label{ex_Piterbarg} 
f(t,V,C,0) := -\left((r^{f}_t-r_t)\left(V_t(\varphi)-C_t\right)+(r^{c}_t-r_t)C_t\right), \, t\,\in\,[0,T].
\end{equation}
In this case, the integral representation \eqref{eq:GBSDEint} of $V$ is of the form
\begin{align}
V_t(\varphi) = B^r_t\Excond{\QQ}{\int_{(t,T]}\frac{dA_u}{B^r_u}+\int_t^{T}\frac{f(u,V,C,0)}{B^r_u}du}{\cF_t}, \, t\,\in\,[0,T].
\end{align} 
If we set $r_t=r^f_t$ $d\PP \otimes dt$-a.s. then we obtain by   \eqref{ex_Piterbarg} that 
\begin{align}
V_t(\varphi)=B^{r^f}_t\Excond{\QQ}{\int_{(t,T]}\frac{dA_u}{B^{r^f}_u}+\int_t^{T}(r^{f}_u-r^c_u)\frac{C_u}{B^{r^f}_u}du}{\cF_t}, 
\, t\,\in\,[0,T].
\end{align}
This corresponds to equation (3) in \cite{pit10}. If we set $r_t=r^c_t$ $d\PP \otimes dt$-a.s. in \eqref{ex_Piterbarg}, we obtain
\begin{align}
V_t(\varphi)=B^{r^c}_t\Excond{\QQ}{\int_{(t,T]}\frac{dA_u}{B^{r^c}_u}-\int_t^{T}(r^{f}_u-r^c_u)\frac{\left(V_t(\varphi)-C_u\right)}{B^{r^c}_u}du}{\cF_t}, \, t\,\in\,[0,T], 
\end{align}
which corresponds to equation (5) in \cite{pit10}.
\end{example}

From Assumption~\ref{assumptionCloseOut} we have $\cV=\hat{V}$. Since $\hat{V}$ is an $\FF$-adapted \textcolor{black}{\cadlag} process, we know from \textcolor{black}{\mbox{\cite[Lemma 2.1]{crepey2015b}}} that $\Delta \hat{V}_\tau=0$ 
\textcolor{black}{
under the Hypothesis \ref{hp:H} between $\mathbb{F}$ and $\mathbb{G}$.} 
The same argument holds true for the collateral process $C,$ which we assumed to be a Lipschitz function of the clean value, and for the initial margin $I$, be it posted or received. 
\textcolor{black}{From now on we identify the clean value $\hat{V},$ the collateral process $C$ and the initial margin $I$ with their left limits.} 
We set 
\begin{align}
\label{eq:redefineEAD}
\begin{aligned}
\theta^C_t&:=(1-R^C)\left(\hat{V}_t-C_t+I^{FC}_t\right)^-, \\
\theta^B_t&:=(1-R^B)\left(\hat{V}_t-C_t-I^{TC}_t\right)^+, \\ 
\end{aligned}
\end{align}

\begin{definition}
We call \emph{pre-default xVA-BSDE} the following $\FF$-BSDE on $[0,T]$ with null terminal condition in $T:$ 
\begin{align} \label{eq:XVApreDef}
\begin{cases}
-d\overline{XVA}_t = \bar{f}(\hat{V}_t-\overline{XVA}_t)dt-\sum_{k=1}^d\bar{Z}^k_tdW^{k,\QQ}_t \\
\overline{XVA}_T = 0, 
\end{cases}
\end{align}
where 
\begin{align}\label{eq:driverXVApreDef} 
	\bar{f}(\hat{V}_t-\overline{XVA}_t)& := -f(t,\hat{V}-\overline{XVA},C,I^{TC},I^{FC}) - 
	(r_t+\lambda^{C,\QQ}_t+\lambda^{B,\QQ}_t)\overline{XVA}_t 
	- \lambda^{C,\QQ}_t\theta^C_t+\lambda^{B,\QQ}_t\theta^B_t, 
\end{align}
for $\theta^B,\,\theta^C$ defined as in \eqref{eq:redefineEAD} and $\lambda_t^{j,\QQ} = r_t^f - 
\lambda_u^{j,\PP} - r_t^j, \; t \,\in\, [0,T], \, j \,\in\, \{B,C\}.$ 
\end{definition}

We now discuss existence and uniqueness for the solution to \eqref{eq:XVApreDef}. \textcolor{black}{In the following, we use subscripts of the form $_{t:T}$ to indicate the dependence on the path of a process from $t$ up to $T$.}
First, we observe that the driver \eqref{eq:driverXVApreDef} also depends on the initial margins. We set 
\begin{equation} \label{def:IM}
I^i_{\textcolor{black}{s}} := \rho_{\textcolor{black}{s}}(\hat{V}_{{\textcolor{black}{s}}:T} 
- \overline{XVA}_{{\textcolor{black}{s}}:T})_{{\textcolor{black}{s}} \,\in\, [t,T]}, \; i \,\in\, \{TC, \, FC \}, 
\end{equation}
where $\overline{XVA}_{t:T} := (\overline{XVA}_s)_{s \,\in\, [t,T]}$ is the process defining the pre-default value adjustment 
and $\hat{V}_{t:T} := (\hat{V}_s)_{s \,\in\, [t,T]}$ is the clean value assumed to be a given exogenous process in $\cS^2(\QQ),$ 
both evaluated up to the contract's maturity, since they are used to measure the potential future exposure. \textcolor{black}{For each 
$t \,\in\, [0.T],$} 
$\rho_t = \rho(\omega,t;x)$, is a risk measure. 
For the sake of simplicity, we assume the same $\rho$ for both $i \,\in\, \{TC,FC\}$. This hypothesis can be easily generalized. 

We also assume the following 
\begin{assumption} \label{ass:A}  
\begin{enumerate}
	\item For any $X, \, Y \,\in\, \cS^2(\QQ),$ the process $(\rho_{\textcolor{black}{t}} (X_{t:T} - Y_{t:T}))_{\textcolor{black}{t}\,\in\,[0,T]}$ is in $\cH^{2\textcolor{black}{,1}}(\QQ).$ There exists a constant 	\textcolor{black}{$C_{\rho} > 0$} and a family of measures $(\nu_s)_{s\,\in\,[0,T]}$ on $\RR$ such that 
				$\nu_t([t; T]) = 1,$ for every $t\,\in\,[0, T],$ and, for any $x,\, y^1,\, y^2 \,\in\, \cS^2(\QQ),$ we have 
				\begin{equation}\label{eq:Lip_IM}
					\big|\rho_t(x_{t:T} - y^1_{t:T}) - \rho_t(x_{t:Y} - y^2_{t:T}) \big| \leq C_{\rho} \EE\left[\int_t^T |y^1_s - y^2_s| 
					\nu_t(ds) \big| \mathcal{F}_t \right], \; dt\otimes d\PP \mbox{ a.e. }
				\end{equation}
	\item we assume that $f$ satisfies Assumption \ref{ass:S_gobet} with  
				$y = \hat{V} - \overline{XVA}^i, \, z = C$ and $\lambda = \rho^i,$ for $i = 1,2.$ 
\end{enumerate} 
\end{assumption} 

\begin{lemma} \label{lemma:aprioriEst} 
Let $(\overline{XVA}^i,\overline{Z}^i) \,\in\, \cS^2(\QQ) \times \cH^{2,d}(\QQ), \, i=1,2,$ be solutions to the $\FF$-BSDE 
\eqref{eq:XVApreDef}, with $f^i(t,\hat{V} - \overline{XVA}^i, C, \rho^i)$ 
satisfying Assumption \ref{ass:A}, for $i=1,2.$  \\
Moreover, define 
\begin{align*}
& \delta\overline{XVA} := \overline{XVA}^1 - \overline{XVA}^2, \\ 
& \delta \overline{Z} := \overline{Z}^1 - \overline{Z}^2, \\ 
& \delta f := f^1(\hat{V} - \overline{XVA}^2) - f^2(\hat{V} - \overline{XVA}^2), \\ 
& \delta \rho := \rho_s^1(\hat{V}_{s:T} - \overline{XVA}_{s:T}^2) - \rho_s^2(\hat{V}_{s:T} - \overline{XVA}_{s:T}^2),  
\end{align*}
for any $s \,\in\, [t,T].$ Then, there exists a constant $\textcolor{black}{c} >0$ such that for $\mu >0$, we have for $\beta$ large enough
\begin{align} \label{eq:apriori1}
\|\delta \overline{XVA}\|_{\cS^{2}_{\beta,T}}^{2} & \leq \textcolor{black}{c} \left[e^{\beta T} \mathbb{E}\left[|\delta \xi|^{2}\right] + 
\frac{1}{\mu^{2}}\left(\left\|\delta_{2} f\right\|_{\cH^{2,d}_{\beta,T}}^{2} + C_{f^1}\left\|\delta_{2} 
\rho\right\|_{\cH^{2,d}_{\beta,T}}^{2}\right)\right],  \\ \label{eq:apriori2} 
\|\delta \overline{Z}\|_{\cH^{2,d}_{\beta,T}}^{2} & \leq \textcolor{black}{c} \left[e^{\beta T} \mathbb{E}\left[|\delta \xi|^{2}\right] 
+ \frac{1}{\mu^{2}}\left(\left\|\delta_{2} f\right\|_{\cH^{2,d}_{\beta,T}}^{2} + 
C_{f^{1}}|| \delta \rho \|_{\cH^{2,d}_{\beta,T}}^{2}\right)\right], 
\end{align} 
where $\|\delta \overline{XVA}\|_{\cS^{2}_{\beta,T}}$ and $\|\delta \overline{Z}\|_{\cH^{2,d}_{\beta,T}}$ are defined in 
\eqref{def:NormS} and \eqref{def:NormH}, respectively. 
\end{lemma}

\begin{proof} 
First, we observe that $f^i(\hat{V} - \overline{XVA}^i), \, i=1,2,$ given in \eqref{eq:driverXVApreDef}, consists 
of three terms. The first one is the full $\GG$-BSDE driver $f,$ given in \eqref{eq:driver} and expressed in terms of the collateral 
$C,$ which is a Lipschitz function of the clean value by definition, and the (posted/received) initial margin $I,$ which is Lipschitz 
by \eqref{def:IM} and \eqref{eq:Lip_IM}. 
The second term depends on the short rate $r$ and the jump intensities $\lambda^{B,\QQ}, \, \lambda^{C,\QQ},$ which are bounded 
by definition. The last term relies upon the close-out conditions $\theta^B, \, \theta^C$ given in \eqref{eq:redefineEAD}, 
which are Lipschitz functions, by following the same arguments as before. Therefore, the driver satisfies Assumption \ref{ass:S_gobet}. 
Moreover, we observe that, for any $(\overline{XVA}^1,\overline{Z}^1), \, (\overline{XVA}^1,\overline{Z}^1) \,\in\, \cS^2(\QQ) 
\times \cH^{2,d}(\QQ),$ we have 
\begin{align*}
|f^1(\hat{V} & - \overline{XVA}^1) - f^2(\hat{V} - \overline{XVA}^2)| \\ 
& \leq |f^1(\hat{V} - \overline{XVA}^1) - f^1(\hat{V} - \overline{XVA}^2)| + |f^1(\hat{V} - \overline{XVA}^2) - f^2(\hat{V} - \overline{XVA}^2)| \\ 
& \leq |f^1(s,\hat{V} - \overline{XVA}^1,C, \rho^1(\hat{V}_{s:T} - \overline{XVA}_{s:T}^1)) 
- f^1(s,\hat{V} - \overline{XVA}^2,C, \rho^1(\hat{V}_{s:T} - \overline{XVA}_{s:T}^2))| \\ 
& + |(r_s + \lambda_s^{C,\QQ} + \lambda_s^{B,\QQ}) (\overline{XVA}^1 - \overline{XVA}^2)| + |\delta f| \\ 
& \leq  (\bar{C}_{f^1} + \bar{C}_s) \left( |\delta\overline{XVA}|\right)
+ \bar{C}_{f^1} \left(|\rho^1(\hat{V}_{s:T} - \overline{XVA}_{s:T}^1) - \rho^1(\hat{V}_{s:T} - \overline{XVA}_{s:T}^2)| + |\delta \rho| \right) + |\delta f|
\end{align*} 
where the last inequality holds true thanks to Assumption \ref{ass:S_gobet}, adding and subtracting 
$\rho^1(\hat{V}_{s:T} - \overline{XVA}_{s:T}^2)$ and choosing a suitable constant 
$\bar{C}_s \geq r_s + \lambda_s^{C,\QQ} + \lambda_s^{B,\QQ}.$ 

\textcolor{black}{The result follows by applying the argument provided in the proof of \cite[Lemma 3.1]{Crpey2020} with terminal condition 
$\xi = 0.$
}
\end{proof}
\begin{proposition} \label{prop:XVApreDef}
Under Assumptions \ref{ass:proxA}, \ref{ass:S_gobet} and \ref{ass:A}, the $\FF$-BSDE \eqref{eq:XVApreDef} is well posed and 
has a unique solution $(\overline{XVA},\overline{Z}) \,\in\, \cS^2(\QQ) \times \cH^{2,d}(\QQ).$ 
\end{proposition} 

\begin{proof} 
\textcolor{black}{
The proof is analogous to the one in \mbox{\cite[Theorem 3.1]{Crpey2020}} by using the estimates \textcolor{black}{\eqref{eq:apriori1} and \eqref{eq:apriori2}}}. 
\end{proof}

Now, given the uniqueness of the solution to \eqref{eq:XVApreDef} we have the following result.

\textcolor{black}{
\begin{proposition} \label{prop:XVAbar}
Let $\left(\overline{XVA}, \bar{Z}\right)$ be the unique solution of the pre-default XVA-BSDE \eqref{eq:XVApreDef}. Define
\begin{align} \label{eq:XVAsolutionFromFtoG}
		X_t := \overline{XVA}_t J_t + \Ind{\tau \leq t}\vartheta_\tau, \ t\in[0,\tau \wedge T], 
\end{align}
where $J_t:=\Ind{t<\tau} = 1-H_t$ and $\vartheta_t := -\theta_t + \hat{V}_t, \, t \,\in\, [0,T].$ Then, under Under Assumptions \ref{ass:proxA} and \ref{ass:A}, the process $\left(X, \bar{Z}, \tilde{U} \right)$ solves the $\GG$-BSDE on $\{\tau > t\}$
\begin{align} \label{eq:XVAGBSDE}
		\begin{cases}
			-dX_t & = -\left[f(t,\hat{V}-XVA,C,I^{TC},I^{FC}) + r_tXVA_t \right]dt 
			- \sum_{k=1}^d\tilde{Z}^k_tdW^{k,\QQ}_t-\sum_{j\in\{B,C\}}\tilde{U}^j_tdM^{j,\QQ}_t\\
			X_\tau & = 
			\left(\hat{V}_\tau(\varphi)-\theta_\tau(\hat{V},C,I^{TC},I^{FC})\right), 
		\end{cases}
\end{align}
where $\tilde{U}^j, \, j \,\in\, \{B,C\},$ are $\GG$-adapted processes in $\cH^{2,2}_{\lambda}(\QQ)$ such that 
\begin{align} \label{eq:pureJumpMartingale}
		\sum_{j\in\{B,C\}}\tilde{U}^j_tdM^{j,\QQ}_t & = -\left((\vartheta_t-\overline{XVA}_t)dJ_t + 
		\lambda^{C,\QQ}_t(-\theta^C_t-\overline{XVA}_t)dt + \lambda^{B,\QQ}_t(\theta^B_t-\overline{XVA}_t)dt\right).  
\end{align}
In particular, $X_t = XVA_t, \, t \,\in\, [0,T],$ where $XVA$ is introduced in Definition~\ref{def:XVA}.
\end{proposition}
}

\begin{proof} 
We start from \eqref{eq:XVAsolutionFromFtoG} and apply the product rule, hence
\begin{align*}
dX_t & = d\left(\overline{XVA}_tJ_t\right)+d\left(\Ind{\tau<T}\vartheta_\tau\right)\\
 & = d\overline{XVA}_{t\wedge \tau}+\overline{XVA}_tdJ_t-\vartheta_tdJ_t. 
\end{align*}
By \eqref{eq:XVApreDef} we obtain 
\begin{align*}
	dX_t & = \left[f(t,\hat{V}-\overline{XVA},C,I^{TC},I^{FC}) + (r_t+\lambda^{C,\QQ}_t+\lambda^{B,\QQ}_t)\overline{XVA}_t 
	+ \lambda^{C,\QQ}_t\theta^C_t - \lambda^{B,\QQ}_t\theta^B_t\right]dt\\
	&\quad + \sum_{i=1}^d\bar{Z}^k_t\Ind{t<\tau}dW^{k,\QQ}_t-\left(\vartheta_t-\overline{XVA}_t\right)dJ_t. 
\end{align*}
We note that the process $\sum_{i=1}^d\int_0^\cdot\bar{Z}^k_u\Ind{u<\tau}dW^{k,\QQ}_u$ is a $(\GG,\QQ)$-martingale, since $\bar{Z}$ 
is in $\cH^{2,d}(\QQ)$ due to the immersion hypothesis. 
From Lemma 5.2.9 in \cite{crepeyBook14} we deduce that the process, expressed in differential form
\begin{align}
\label{eq:jumpMart}
-\left((\vartheta_t-\overline{XVA}_t)dJ_t+\lambda^{C,\QQ}_t(-\theta^C_t-\overline{XVA}_t)dt 
+ \lambda^{B,\QQ}_t(\theta^B_t-\overline{XVA}_t)dt\right)
\end{align}
is also a $(\GG,\QQ)$-local martingale. 
Moreover, we observe that, since $\hat{V}\,\in\,\cS^2(\QQ),$ also $C \,\in\,\cS^2(\QQ),$ $C$ being a Lipschitz function of 
$\hat{V}.$ Additionally, the initial margin, be it posted or received, lies in $\cH^2(\QQ)$ by assumption.  
Summing up, both $\theta^B$ and $\theta^C$, and hence $\vartheta$ belong to the space $\cH^2(\QQ)$. 
On the other hand, $\overline{XVA}\in\cS^2(\QQ)$. Recalling that both $\lambda^{C,\QQ}$ and $\lambda^{B,\QQ}$ are bounded, it 
follows that the compensated jump term \eqref{eq:jumpMart} is a square integrable martingale.
\textcolor{black}{
Since the predictable representation property holds in $\GG$ with respect to $W^{j,\QQ}, \, M^{j,\QQ}, \, j \,\in\, \{B,C\},$ under the measure $\QQ,$ we obtain that there exists $\tilde{U}^B, \, \tilde{U}^C \,\in\, \cH^{2,2}_{\lambda}(\QQ)$ satisfying \eqref{eq:pureJumpMartingale}. 
}
We conclude that the process $XVA$ solves the xVA-BSDE \eqref{eq:XVAGBSDE} under the filtration $\GG$. 
\end{proof}

We can finally combine the solution of the BSDE \eqref{eq:XVApreDef} for the clean value with the result \textcolor{black}{of Proposition \ref{prop:XVAbar}} to solve 
the $\GG$-BSDE \eqref{eq:GBSDEdiff}.

\begin{theorem} \label{thm:solGbsde}
Let $V_t := \hat{V}_t - XVA_t, \, t \,\in\, [0,T],$ on $\{\tau > t \},$ 
where $\hat{V}$ and $XVA$ are defined in \eqref{fair_price:clean_mkt} and \eqref{eq:XVA}, respectively. 
Then, under Assumptions \ref{ass:proxA} and \ref{ass:A}, the triplet $\left(V, Z, U \right) \,\in\, \cS^{2}(\QQ) 
\times \cH^{2,d}(\QQ) \times \cH^{2,2}_{\textcolor{black}{\lambda}}(\QQ)$ 
solves the $\GG$-BSDE \eqref{eq:GBSDEdiff} \textcolor{black}{with $\mathcal{V} = \hat{V}$}, where $Z$ and $U$ are given by 
\begin{align} \label{def:Z}
Z_t^{k} = \hat{Z}_t^k - \tilde{Z}_t^k, \, k = 1, \ldots, d, \\ \label{def:U} 
U_t^{j} = - \tilde{U}_t^{j}, \, j \,\in\, \{B,C \}. 
\end{align}
Moreover, the process $V$ satisfies \eqref{eq:GBSDEint}.
\end{theorem}

\begin{proof}
By Hypothesis \ref{hp:H}, on $\{t<\tau\}$ we have 
\begin{align*}
\hat{V}_t(\varphi)=B^r_t\Excond{\QQ}{\int_{(t,T]}\frac{dA_u}{B^r_u}}{\cF_t}=B^r_t\Excond{\QQ}{\int_{(t,T]}\frac{dA_u}{B^r_u}}{\cG_t}.
\end{align*}
So we consider $\hat{V}$ under $\GG$. We also observe that, on  $\{t < \tau\}$, we have $\bar{A}_t=A_t.$ 
Using \eqref{cleanBSDE} and \eqref{eq:XVAGBSDE}, we write the dynamics of $V$ on $\{t < \tau\}$ 
\begin{align*} 
-dV_t & = d \bar{A}_t + \left[f(t,\hat{V}-XVA,C,I^{TC},I^{FC})dt - r_t \left(\hat{V}_t - XVA_t \right) \right]dt \\ 
 & - \sum_{k=1}^d\left(\hat{Z}^k_t - \tilde{Z}^k_t\right) dW^{k,\QQ}_t - \sum_{j\in\{B,C\}}\left(-\tilde{U}^j_t\right) 
 dM^{j,\QQ}_t
\end{align*}
with terminal condition at $\tau$
\begin{align*}
V_\tau & = \hat{V}_{\tau} - XVA_{\tau} = \hat{V}_{\tau} - \left(\theta_{\tau} - \hat{V}_{\tau} \right) = \theta_{\tau}.
\end{align*}
Since $Z = \hat{Z} - \tilde{Z} \,\in\, \cH^{2,d}(\QQ)$ and $U = - \tilde{U} \,\in\, \textcolor{black}{\cH^{2,2}_{\lambda}(\QQ)},$ by Theorem~\ref{thm:solFbsde} and Proposition~\ref{prop:XVAbar}, we obtain that $\left(V, Z, U \right)$ solves the 
$\GG$-BSDE \eqref{eq:GBSDEdiff} and satisfies the required integrability conditions. 

Finally, we are now able to prove that \eqref{eq:GBSDEint} is equivalent to \eqref{eq:GBSDEdiff}.  

Here we assume to work only on $\{ \tau > t\}.$ Since $V_t = \hat{V} - XVA_t$ and thanks to Definition~\ref{def:XVA} we have  
\begin{align*} 
V_t(\varphi) = \hat{V}_t(\varphi) & + B^r_t\mathbb{E}^{\mathbb{Q}}\left[
\Ind{\tau^C< \tau^B} \frac{(1-R^C)\left(\hat{V}_\tau(\varphi)-C_{\tau-} + I^{FC}_{\tau-}\right)^-}{B^r_\tau} \right. \\
& \left. - 
\Ind{\tau^B< \tau^C}\frac{(1-R^B)\left(\hat{V}_\tau(\varphi)-C_{\tau-} 
- I^{TC}_{\tau-}\right)^+}{B^r_\tau} \right. \\ 
& \left. - \int_t^{\textcolor{black}{\tau}} 
\frac{(r^{f,l}_u-r_u)\left(V_u(\varphi)-C_u-I^{TC}_u\right)^+ 
- (r^{f,b}_u-r_u)\left(V_u(\varphi)-C_u-I^{TC}_u\right)^-}{B^r_u}du \right. \\ 
& \left. - \int_t^{\textcolor{black}{\tau}} 
\frac{(r^{c,l}_u-r_u)C^+_u-(r^{c,b}_u-r_u)C^-_u}{B^r_u}du 
- \int_t^{\textcolor{black}{\tau}} 
\frac{(r^{I,l}_u-r_u)I^{TC}_u-r^{I,b}_uI^{FC}_u}{B^r_u}du \Big|\cG_t\right]. 
\end{align*}

By \eqref{eq:driver} we obtain 
\begin{align*} 
V_t(\varphi) = \hat{V}_t(\varphi) & + B^r_t\mathbb{E}^{\mathbb{Q}}\left[
\Ind{\tau^C< \tau^B} 
\frac{(1-R^C)\left(\hat{V}_\tau(\varphi)-C_{\tau-} + I^{FC}_{\tau-}\right)^-}{B^r_\tau} \right. \\
& \left. - 
\Ind{\tau^B< \tau^C}\frac{(1-R^B)\left(\hat{V}_\tau(\varphi)-C_{\tau-} 
- I^{TC}_{\tau-}\right)^+}{B^r_\tau} \right. \\ 
& \left. + \int_t^{\textcolor{black}{\tau}} 
\frac{f(u,V,C,I^{TC},I^{FC})}{B^r_u}du \Big|\cG_t\right]. 
\end{align*} 

Assumption~\ref{assumptionCloseOut} and \eqref{A_bar} ensure that 
\begin{align*} 
V_t(\varphi) = \hat{V}_t(\varphi) & + B^r_t\mathbb{E}^{\mathbb{Q}}\left[
\int_t^{\textcolor{black}{\tau}} 
\frac{f(u,V,C,I^{TC},I^{FC})}{B^r_u}du + 
\frac{\theta_{\tau}(\hat{V}(\varphi),C,I^{TC},I^{FC}) 
- \hat{V}_{\tau}(\varphi)}{B^r_\tau} \Big|\cG_t\right]. 
\end{align*} 
Now, we apply \eqref{fair_price:clean_mkt}, the tower property and Hypothesis~\ref{hp:H}, so that 
\begin{align*} 
V_t(\varphi) & = B^r_t\mathbb{E}^{\mathbb{Q}}\left[\int_t^{\textcolor{black}{\tau}} 
\frac{f(u,V,C,I^{TC},I^{FC})}{B^r_u}du 
+ \Ind{\tau \leq T}\frac{\theta_{\tau}(\hat{V}(\varphi),C,I^{TC},I^{FC})}{B^r_\tau} \Big|\cG_t\right] \\ 
& + B^r_t\mathbb{E}^{\mathbb{Q}}\left[\int_{(t,T]} \frac{d A_u}{B_u^r} - 
\int_{(t,T]} \frac{d A_u}{B^r_\tau} 
\Big|\cG_t\right]. 
\end{align*} 
Finally, again by \eqref{A_bar}, we have 
\begin{align*}
V_t(\varphi) = B^r_t\mathbb{E}^{\mathbb{Q}}\left[ \int_{(t,T]} \frac{d \bar{A}_u}{B_u^r} + \int_t^{\textcolor{black}{\tau}} 
\frac{f(u,V,C,I^{TC},I^{FC})}{B^r_u}du + 
\frac{\theta_{\tau}(\hat{V}(\varphi),C,I^{TC},I^{FC})}{B^r_\tau} \Big|\cG_t\right].  
\end{align*}
\end{proof}

We now provide an explicit formula for the value adjustments under the filtration $\FF$. 
This representation is particularly useful from a computational point of view: risk factors can be simulated under 
the smaller filtration $\FF$ and the computation of value adjustments does not require the simulation of default times. 
It is an immediate consequence of Proposition~\ref{prop:XVApreDef}.

\begin{corollary}\label{eq:corollary}
Let $\left(\overline{XVA}, \bar{Z}\right)$ be the unique solution to 
the pre-default xVA-BSDE under $\FF$ \eqref{eq:XVApreDef}. Define the process $\tilde{r}=\left(\tilde{r}_t\right)_{t\in[0,T]}$ 
by setting $\tilde{r}:=r+\lambda^{C,\QQ}+\lambda^{B,\QQ}$. Under Assumptions \ref{ass:proxA} and \ref{ass:A} the stochastic process 
$\overline{XVA}$ admits the following representation.
\begin{align}\label{def:XVAPredef}
\overline{XVA}_t = -\overline{CVA}_t + \overline{DVA}_t + \overline{FVA}_t + \overline{ColVA}_t + \overline{MVA}_t,  
\end{align}
where 
\begin{align*}
\overline{CVA}_t& := B^{\tilde{r}}_t\Excond{\QQ}{(1-R^C)\int_t^{T}\frac{1}{B^{\tilde{r}}_u}\left(\hat{V}_u(\varphi)-C_{u} 
+ I^{FC}_{u}\right)^-\lambda^{C,\QQ}_udu}{\cF_t}, \\
\overline{DVA}_t& := B^{\tilde{r}}_t\Excond{\QQ}{(1-R^B)\int_t^{T}\frac{1}{B^{\tilde{r}}_u}\left(\hat{V}_u(\varphi)-C_{u} 
- I^{TC}_{u}\right)^+\lambda^{B,\QQ}_u du}{\cF_t}, \\
\overline{FVA}_t& :=B^{\tilde{r}}_t \Excond{\QQ}{\int_t^{T}\frac{(r^{f,l}_u-r_u)\left(V_u(\varphi)-C_u-I^{TC}_u\right)^+ 
- (r^{f,b}_u-r_u)\left(V_u(\varphi)-C_u-I^{TC}_u\right)^-}{B^{\tilde{r}}_u}du}{\cF_t}, \\
\overline{ColVA}_t& :=B^{\tilde{r}}_t \Excond{\QQ}{\int_t^{T}\frac{(r^{c,l}_u-r_u)C^+_u-(r^{c,b}_u-r_u)C^-_u}{B^{\tilde{r}}_u}du}{\cF_t}, \\
\overline{MVA}_t& := B^{\tilde{r}}_t\Excond{\QQ}{\int_t^{T}\frac{(r^{I,l}_u-r_u)I^{TC}_u-r^{I,b}_uI^{FC}_u}{B^{\tilde{r}}_u}du}{\cF_t}.
\end{align*}
\end{corollary}

In the literature there has been an intense debate regarding the problem given by the possible overlap between FVA and DVA, see e.g. \cite{hw2012}, \cite{ads2019}, \cite{bfp2019} and references therein. This problem, however, is due to accounting inconsistencies that do not affect our pricing equations.  We limit ourselves to mention that a sound treatment of the issue is provided by \cite{bfp2019} and that their solution can be 
embedded in our framework at the cost of further notations.

\subsection{\textcolor{black}{The xVA-CSA consistency problem}}\label{Sect:cons_prob}
\textcolor{black}{We here provide an arbitrage-free framework for the CSA discounting practice, i.e. the practice of using contingent claim specific discounting regimes.}  

In Section \ref{sect:clean} we assumed that the clean value refers to an idealized fully collateralized transaction where the 
collateral rate is simply $r$. The situation in practice is more complicated. The market practice adopted for the computation 
of clean prices involves a multitude of discount curves. Possible examples from the \textcolor{black}{current market practice} are 
\begin{itemize}
\item The (clean) value of a perfectly uncollateralized derivative might be discounted by a bank by means of a bank-specific 
			funding curve with associated short rate $r^f$ (this could correspond to the Libor rate for a bank belonging to the Libor 
			panel), see e.g. \cite{pit10}.
\item The (clean) value of a derivative collateralized in a foreign currency is discounted on the market at a rate depending on 
			cross currency bases, see the formulas and derivations in \cite{gnoatto2020}. 
\end{itemize}

It is quite natural to ask why banks employ \textcolor{black}{many} discount regimes for clean values and, on top of that, xVA corrections. 
The main reason is purely pragmatic and non-mathematical: from the perspective of a trading desk it is convenient to treat CSAs by means of different discount regimes, because this allows to deal with portfolio market risk via traditional 
trading-desk techniques, such as curve trades (i.e. e.g. buying/selling interest rate swaps on different buckets/maturities 
along the curve). Hedging the expectation of an integral such as the FVA term in practice is much more complicated. 
A possible approximate treatment involves discretizing the time integral and treating the resulting Riemann sum over time 
as a portfolio of claims. In view of the aforementioned difficulty, market operators prefer to obtain an additive price 
representation, where discount curves are used to reduce the magnitude of the (funding related) xVA terms, which are more 
difficult to hedge.

From now on, we shall assume that the bank has two internal desks, dubbed the \emph{front-office desk} and the \emph{xVA desk,} 
respectively. The front-office desk is responsible for the calculation of clean values and for the trading activity required to 
hedge market risk of the clean values. The xVA desk instead computes and hedges all the value adjustments and is forced, 
according to internal rules of the bank, to adopt for each transaction the clean value dictated by the front-office desk. 
The fact that the xVA desk is a \textit{clean-value-taker} implies that care is needed when computing xVAs, in order to avoid 
double counting effects. 

The xVA desk has to deal with two different clean values for the same transaction. 
On the one side, \textcolor{black}{it performs an arbitrage-free pricing by computing the clean value by means of the $\FF$-BSDE \eqref{cleanBSDE}.}
On the other side, \textcolor{black}{it has to use the} clean value prescribed by the front office function, which constitutes the official clean value 
accepted within the bank. The xVA desk is then faced with the following challenge:

\begin{problem}[xVA-CSA consistency problem] \label{prob:consistency}
Produce a price decomposition of $V$ in terms of clean value and xVA such that
\begin{itemize}
\item[\emph{(i)}] the representation of $V$ is coherent with the $\GG$-BSDE \eqref{eq:GBSDEdiff}, and 
\item[\emph{(ii)}] the clean price corresponds to the one prescribed by the front-office function.
\end{itemize}
\end{problem}

We now provide a solution for Problem \ref{prob:consistency} by using \textcolor{black}{the} results of Section \ref{Sect:single}. 

\begin{remark}\label{rem:ass_multi}
To provide a concrete example, 
\textcolor{black}{consider} the situation where the trading desk of the 
bank enters into two perfectly collateralized transactions with two different counterparties, the first one being e.g. a 
clearing house such as LCH, the other one being another clearing house such as Eurex. Although the dividend process of the 
claim is the same for both transactions, the collateral remuneration provided by the trade with Eurex and the trade with 
LCH is different. The spread in the collateral remuneration between EUREX and LCH is called \emph{Eurex-LCH basis,} see 
e.g. \cite{rms17} for a more detailed discussion. 
This will result in the two clean values being computed by means of different discounting rates.
\end{remark} 

In summary, the market practice of discounting cashflows according to trade-specific collateral rates implies that, within the bank, 
a single transaction will be discounted at least according to two different regimes. 
Initially, the front-office determines the clean value by discounting cash flows through an ideal market collateral rate $\hat{r}.$ 
Hence the front-office clean value $\textcolor{black}{\hat{\mathcal{P}} = (\hat{\mathcal{P}}_t)_{t \,\in\, [0,T]}}$ is obtained from the $\FF$-BSDE
\begin{align}\label{ex:price}
\begin{cases}
-d\textcolor{black}{\hat{\mathcal{P}}}_t = -\sum_{k=1}^d \textcolor{black}{\hat{Z}^{*,k}_t} d W^{k,\QQ}_t + d A_t - \hat{r} \textcolor{black}{\hat{\mathcal{P}}}_t dt, \\
\textcolor{black}{\hat{\mathcal{P}}}_T = 0.  
\end{cases}
\end{align}
On the other side, the xVA desk first computes the clean value $\hat{V} = (\hat{V}_t)_{t \,\in\, [0,T]}$ as the solution to the $\FF$-BSDE \eqref{cleanBSDE}, 
i.e. by solving
\begin{equation}\label{end:price} 
\begin{cases}
-d\hat{V}_t = -\sum_{k=1}^d\hat{Z}^{k}_t dW^{k,\QQ}_t  + d A_t - r_t\hat{V}_t dt, \\
\hat{V}_T = 0.   
\end{cases}
\end{equation}
From a valuation perspective, if clean values represented the prices of real transactions, the presence of \textcolor{black}{different} discounting 
rules would immediately imply the presence of trivial arbitrage opportunities in the market. Only the endogenous price 
\textcolor{black}{$\hat{V}$ in} \eqref{end:price} 
is compatible with the arbitrage-free setting of Section~\ref{Sect:single}. On the other hand, the xVA desk is forced to provide results 
in terms of the discounting regimes imposed by the front-office. 
The two approaches can be combined in an arbitrage-free setting by means of the following invariance property of linear BSDEs. 

\begin{lemma}\label{lemmaFundingInvariance}
Let $(\hat{V}, \textcolor{black}{\hat{Z}^{1},\ldots, \hat{Z}^{d}})$ be the unique solution of the $\FF$-BSDE \eqref{end:price}. 
Under Assumption \ref{cleanMarket} for $A,$ the value process $\hat{V}$ admits the two equivalent 
representations
\begin{itemize}
\item[i)] xVA-discounting representation
\begin{align}
\label{XVADiscCleanValue}
\hat{V}_t = B^{r}_t \EE^{\QQ}\left[ \int_{(t,T]} \frac{d A_t}{B^{r}_u} \Big| \cF_t \right], 
\end{align}
\item[ii)] CSA-discounting representation
\begin{equation}\label{end:price2} 
\hat{V}_t = \textcolor{black}{\hat{\mathcal{P}}}_t - DiscVA_t, 
\end{equation}
where $DiscVA_t$ represents the \emph{discounting valuation adjustment,} defined as
\begin{equation}\label{discVA:def}
DiscVA_t := B_t^{\hat{r}} \EE^{\QQ}\left[ \int_t^T (r_u - \hat{r}) \frac{\hat{V}_u}{B_u^{\hat{r}}} du \Big| \cF_t \right],
\end{equation}
and $\textcolor{black}{\hat{\mathcal{P}}}$ is the value process in the solution $(\textcolor{black}{\hat{\mathcal{P}}}, \hat{Z}^{*,1},\ldots, \hat{Z}^{*,d})$ of the $\FF$-BSDE \eqref{ex:price}
\begin{equation}
\textcolor{black}{\hat{\mathcal{P}}}_t = B_t^{\hat{r}} \EE^{\QQ}\left[ \int_{(t,T]} \frac{d A_t}{B_u^{\hat{r}}} \Big| \cF_t \right]. 
\end{equation}
\end{itemize}
\end{lemma}

\begin{proof}
The integral representation \eqref{XVADiscCleanValue} is immediate. 
To obtain \eqref{end:price2} we rewrite the $\FF$-BSDE \eqref{end:price} adding and subtracting the term 
$\hat{r} \hat{V}_t,$ i.e., 
\begin{align*}
\begin{cases}
-d\hat{V}_t = -\sum_{k=1}^d\hat{Z}^{k}_t dW^{k,\QQ}_t  + d A_t - (r_t-{\hat{r}}_t)\hat{V}_t dt 
- {\hat{r}}_t\hat{V}_t dt \\
\hat{V}_T = 0.
\end{cases}
\end{align*}
The value process of the solution is given by
\begin{align*}
\hat{V}_t = B_t^{\hat{r}} \EE^{\QQ}\left[ \int_{(t,T]} \frac{d A_t}{B_u^{\hat{r}}} \Big| \cF_t \right] 
- B_t^{\hat{r}} \EE^{\QQ}\left[ \int_t^T (r_u - {\hat{r}}_u) \frac{\hat{V}_u}{B_u^{\hat{r}}} du \Big| \cF_t \right], 
\end{align*}
where we recognize the first expectation as $\hat{P}$, whereas the second one provides $DiscVA$.
\end{proof}
The $\GG$-BSDE \textcolor{black}{for the dynamics of $V$ under $\mathbb{Q}$}
can be written as
\begin{align}\label{fullBSDE:multi}
\begin{cases}
-d V_t(\varphi) & = d \bar{A}_t + \left(f(t,V,C,I^{TC},I^{FC}) - r_t V_t(\varphi) \right) dt \\
&\quad- \sum_{k=1}^d Z_t^k d W_t^{k,\QQ} - \sum_{j\,\in\,\{B,C\}} U_t^j dM_t^{j,\QQ}, \\ 
V_{\tau}(\varphi) & = \theta_{\tau}\left(\hat{V},C,I\right),  
\end{cases}
\end{align}
\textcolor{black}{for $t \leq \tau,$}
where $\bar{A}_t$ is defined in \eqref{A_bar}, $Z_t = \left(Z_t^{1}, \ldots, Z_t^{d} \right),$ 
and $U_t = \left(U_t^{B}, U_t^{C} \right),$ represent the control processes given by 
$\GG$-predictable processes, and $f(t,V,C,I^{TC},I^{FC})$ is the $\GG$-BSDE driver given by \eqref{eq:driver}. 
The close-out condition is 
\begin{align} \label{eq:closeout_multi}
\begin{aligned}
V_{\tau}(\varphi) & = \textcolor{black}{\hat{\mathcal{P}}}_{\tau} + \Ind{\tau^C< \tau^B}(1-R^C)\left(\textcolor{black}{\hat{\mathcal{P}}}_{\tau} - C_{\tau-} 
+ I^{FC}_{\tau-} - DiscVA_{\tau} \right)^-\\
 & - \Ind{\tau^B< \tau^C}(1-R^B)\left( \textcolor{black}{\hat{\mathcal{P}}}_{\tau} - C_{\tau-} - I^{TC}_{\tau-} -  DiscVA_{\tau}\right)^+. 
\end{aligned}
\end{align}

By using the same arguments given for Theorem~\ref{thm:solGbsde} 
and taking into account Definition \ref{def:XVA}, we obtain the following result, with the help of 
Lemma~\ref{lemmaFundingInvariance}.

\begin{proposition}\label{fullGG:multi}
Under Assumption~\ref{cleanMarket} and ~\ref{ass:A}, 
the $\GG$-BSDE \eqref{fullBSDE:multi} \textcolor{black}{admits a unique 
solution $(V,Z,U) \,\in\, \cS^2(\mathbb{Q}) \times \cH^{2,d}(\mathbb{Q}) \times \cH^{2,2}_{\lambda}(\mathbb{Q})$ with} 
\begin{align} 
\label{eq:xvaMultiDiscount}
\begin{aligned}
V_t(\varphi) & =  B_t^{\hat{r}} \mathbb{E}^{\mathbb{Q}}\left[\int_{(t,T]} \frac{d A_u}{B_u^{\hat{r}}} \Big| 
\cF_t \right] - XVA_t =  \textcolor{black}{\hat{\mathcal{P}}}_t -\widehat{XVA}_t,
\end{aligned}
\end{align}
\textcolor{black}{for $t < \tau,$} where 
\begin{align}
\widehat{XVA}:= XVA+DiscVA, 
\end{align}
with
\begin{align} \label{eq:XVA_t}
\begin{aligned}
XVA_t & := FVA_t + ColVA_t + MVA_t - CVA_t +DVA_t\\ 
 & =FVA_t + ColVA_t + MVA_t \\ 
 & - B_t^r \EE^{\QQ}\left[\left. \Ind{\tau < T}\Ind{\tau^C< \tau^B}(1-R^C)\frac{1}{B^r_\tau} 
   \left( \textcolor{black}{\hat{\mathcal{P}}}_{\tau} - C_{\tau-} + I^{FC}_{\tau-} 
   -  DiscVA_{\tau}\right)^- \right| \cG_t\right] \\ 
 & + B_t^r \EE^{\QQ}\left[\left. \Ind{\tau < T}\Ind{\tau^B< \tau^C}(1-R^B)\frac{1}{B^r_\tau} 
\left(\textcolor{black}{\hat{\mathcal{P}}}_{\tau}  -C_{\tau-}+I^{TC}_{\tau-} -  DiscVA_{\tau}\right)^+ \right| \cG_t\right] 
 \end{aligned}
\end{align}
and
\begin{align*}
  DiscVA_t &:=  B_t^{\hat{r}} \EE^{\QQ}\left[ \int_t^T (r_u - {\hat{r}}_u) 
	\frac{\hat{V}_u}{B_u^{\hat{r}}} du \Big| \cF_t \right], 
\end{align*}
where FVA, ColVA, MVA are defined in line with Definition~\ref{def:XVA}. 
\end{proposition}

\textcolor{black}{
\begin{proof} 
Existence and uniqueness for the solution of \eqref{fullBSDE:multi} follow along the lines of Section~\ref{Sect:single}.
\end{proof}
}

\textcolor{black}{
The impact of CSA discounting on derivative exposures is represented by the presence of $\hat{\mathcal{P}}$ in the CVA and DVA terms in \eqref{eq:XVA_t}. Furthermore, decomposition \eqref{eq:XVA_t} shows the exposure profile in the portfolio under $r$ or $\hat{r},$ respectively. 
}

\subsection{\textcolor{black}{xVA multiple aggregation level}} 
\textcolor{black}{In this section we provide a $\GG$-BSDE for the global portfolio of the bank given by the aggregation of several sub-portfolios of trades. 
Current market practice shows that 
}
the set of trades between the bank and the counterparty can be typically split into several subsets reflecting multiple 
aggregation levels. One can distinguish between funding/margin sets and netting sets. 
\textcolor{black}{
A margin (or funding) set $\cM$ is a set of claims whose aggregated clean values (exposures) are fully or partially 
covered by a CSA (collateral agreement). We denote by $N_\cM$ the number of margin sets in a portfolio $\cA$. 
A netting set $\cN$ is a set of margin sets whose post-margin exposures can be aggregated. 
We denote by $N_\cN$ the number of netting sets in a portfolio $\cA$. 
}

\begin{figure}[h]
\begin{tikzpicture}[node distance=1.5cm]
    \node (Item) [abstract, rectangle split, rectangle]
        {
            \textbf{Counterparty}
        };
 
    \node (AuxNode01) [text width=2cm, below=of Item] {};
    
    \node (Component) [abstract, rectangle split, rectangle split parts=2, left=of AuxNode01,yshift = -0.5cm]
        {
            \textbf{Netting Set 1}
            \nodepart{second} e.g. a first subsidiary
        };
        
    \node (System) [abstract, rectangle split, rectangle split parts=2, right=of AuxNode01,yshift = -0.5cm]
        {
            \textbf{Netting Set 2}
            \nodepart{second} e.g. a second subsidiary
        };

    \node (AuxNode02) [abstract, rectangle split, rectangle split parts=2, below=of Component, yshift = -0.5cm]
        {
            \textbf{Margin Set 2}
            \nodepart{second} USD collateralized Trades
        };
    
    \node (Sensor) [abstract, rectangle split, rectangle split parts=2, left=of AuxNode02, xshift = 0.5cm]
        {
            \textbf{Margin Set 1}
            \nodepart{second}Unsecured Trades
        };
        
    \node (Part) [abstract, rectangle split, rectangle split parts=2, right=of AuxNode02, xshift = -0.5cm]
        {
            \textbf{Margin Set 3}
            \nodepart{second} EUR collateralized Trades
        };
        
    \node (AuxNode05) [below=of System] {};
    \node (CoolingSystem) [abstract, rectangle split, rectangle split parts=2, left=of AuxNode05, xshift=2cm, yshift=-1.5cm]
        {
            \textbf{Legacy EUR Margin Set}
            \nodepart{second} Monthly margin calls.
        };
    \node (CoolingLoop) [abstract, rectangle split, rectangle split parts=2, right=of AuxNode05, xshift=-1cm, yshift=-1.5cm]
        {
            \textbf{New EUR Margin Set}
            \nodepart{second} Daily margin calls. 
        };
     
    \draw[myarrow] (Component.north) -- ++(0,0.8) -| (Item.south);
  \draw[line] (Component.north) -- ++(0,0.8) -| (System.north);   
    \draw[myarrow] (Sensor.north) -- ++(0,0.8) -| (Component.south);
  \draw[line] (Sensor.north) -- ++(0,0.8) -| (Part.north); 
    \draw[myarrow] (AuxNode02.north) -- ++(0,0.8) -| (Component.south);
    \draw[myarrow] (CoolingSystem.north) -- ++(0,0.8) -| (System.south);
    \draw[line] (CoolingSystem.north) -- ++(0,0.8) -| (CoolingLoop.north);
          
\end{tikzpicture}
\caption{A possible hierarchical structure of aggregation levels.\label{fig:hierarchy}}
\end{figure}
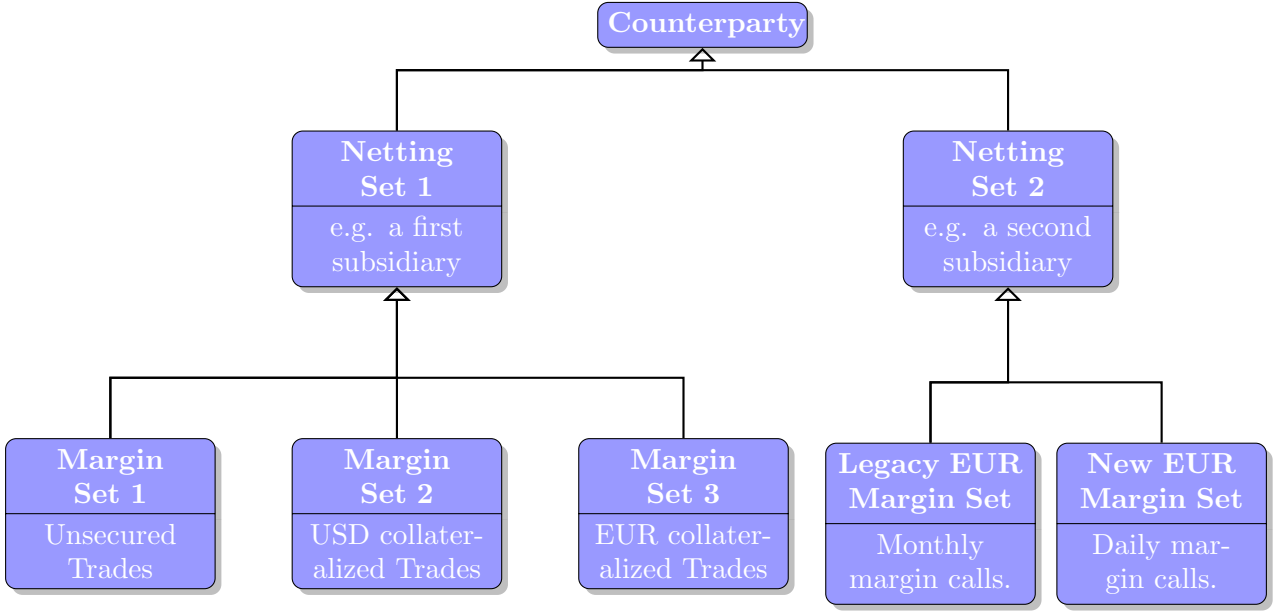 

\begin{example} \label{es:multilevel}
Funding/margin sets are traded between the bank and the 
counterparty that share the same funding policy. This corresponds to different CSAs: for example, 
one CSA (Margin Set 2) could group all trades for which collateral is exchanged in USD (e.g. for foreign exchange derivatives), 
whereas another CSA (Margin Set 3) could be relevant for all instruments collateralized in EUR. 
Finally, trades that are not collateralized, but whose exposures are netted among each other, can be also grouped in a separate 
margin/funding set, corresponding to Margin Set 1 in Figure~\ref{fig:hierarchy}. \\
The protection provided by collateralization agreements might however be imperfect, hence a legal agreement between the bank and the counterparty might allow for the netting of residual post collateral exposures arising from different margin sets. 
This corresponds to Netting Set 1 in Figure~\ref{fig:hierarchy}. \\
Another typical source of multiple aggregation levels is the historical stratification of legal agreements: 
in Figure~\ref{fig:hierarchy} we have a second netting set, corresponding to a second subsidiary of the parent counterparty, 
where legacy trades are covered by an old CSA agreement involving monthly margin calls, whereas all trades entered after a 
certain date are covered by a newer CSA agreement involving daily margin calls. 
\end{example} 
\begin{remark}
A further level of complexity could arise when the parent and the subsidiaries have different default times: 
this introduces further complications when modeling the close-out condition because one might have e.g. a situation 
where the default of a subsidiary is covered by the parent. Such issues are left for future research. 
From a practical point of view it is also difficult to find calibration instruments for default probabilities, 
since subsidiaries typically do not enjoy a liquid CDS market. 
\end{remark} 

We assume that the portfolio $\cA$ of trades between the bank and the counterparty 
consists of $K$ trades, that we identify by means of the respective payment processes, i.e., $\cA=\left\{A^1,\ldots A^K\right\}$. 
We use again $\hat{V}^m$ to denote the clean reference value of the claims $A^m, \, m=1,\ldots K$, representing also their 
credit exposure before collateral is applied. \\
For every claim $A^m\in\cA$, $m=1,\ldots K,$ we assume that the margin set for initial and variation margin coincide. \\
The structure of the portfolio is illustrated in Figure~\ref{fig:hierarchy}, where the first row depicts the composition of 
all margin sets as groups of claims and the second line describes the netting sets as groups of margin sets. Hence,  
\begin{align} \label{eq:ptf_A}
\begin{aligned}
\cA &=  \underbrace{\{A^1,\ldots, A^{N_1} \}}_{\cM_{1}} \cup \underbrace{\{A^{N_1+1}, \ldots, A^{N_2} \}}_{\cM_{2}} \cup \ldots \cup \underbrace{\{A^{N_{N_\cM-1}+1},\ldots, A^{N_{N_\cM}} \}}_{\cM_{N_{\cM}}}\\ 
&=\underbrace{\{\cM_1,\ldots, \cM_{M_{1}} \}}_{\cN_{1}} \cup \underbrace{\{\cM_{M_{1}+1}, \ldots, \cM_{M_{2}} \}}_{\cN_{2}} \cup \ldots \cup  \underbrace{\{\cM_{M_{N_{\cM}-1}+1},\ldots, \cM_{M_{N_{\cM}}} \}}_{\cN_{N_\cN}} 
\end{aligned}
\end{align}
where we have $N_{N_\cM}=K$. \\ 
\textcolor{black}{
We now provide a $\GG$-BSDE for the portfolio at the $m_2$-th netting set-level in \eqref{eq:ptf_A} by using the results of Section \ref{sect:clean}. 
The presence of different margin sets is represented by the introduction of different collateral accounts. Multiple netting sets are accounted for by summing value adjustments over all netting sets, each netting set possibly featuring several margin sets. } \\  
\textcolor{black}{
To this purpose, we resort to a multi-index notation for the processes involved. In particular, we denote by $\bar{A}_t^{m,m_1,m_2}$ 
the $m$-th contingent claim in the portfolio $\mathcal{A}$ belonging to the $m_1$-th netting set-level and at the $m_2$-th margin 
set-level. The same notation applies to the control processes $Z_t^{k,m_1,m_2}$ and $U_t^{j,m_1,m_2},$ for $k = 1, \ldots, d$ and 
$j \,\in\, \{B,C\}.$ 
Hence, the value process $V_t^{m_2}$ corresponding to the $m_2$-th netting set, $m_2 = 1, \ldots, N_{\mathcal{N}},$ satisfies 
the following \textcolor{black}{$m_2$-th netting set-level} BSDE 
}
\begin{align} \nonumber 
-d V_t^{m_2}(\varphi) & = \sum_{m_1=1}^{|\mathcal{N}_{m_2}|} \sum_{m=1}^{|\mathcal{M}_{m_1}|} d \textcolor{black}{\bar{A}}_t^{m,m_1,m_2} \\ \nonumber 
 & + \left[f\left(t, V^{m_2}, (C^{\mathcal{M}_{m_1},\mathcal{N}_{m_2}})_{m_1=1,\ldots,|\mathcal{N}_{m_2}|}, 
(I^{\mathcal{M}_{m_1},\mathcal{N}_{m_2}})_{m_1=1,\ldots,|\mathcal{N}_{m_2}|} \right) - r_t V_t^{m_2}(\varphi) \right] dt \\ 
\label{eq:GBSDE_netting} 
 & - \sum_{m_1=1}^{|\mathcal{N}_{m_2}|} \sum_{m=1}^{|\mathcal{M}_{m_1}|} \sum_{k=1}^d Z_t^{k,m_1,m_2} d W_t^{k,\mathbb{Q}}  
 - \sum_{m_1=1}^{|\mathcal{N}_{m_2}|} \sum_{m=1}^{|\mathcal{M}_{m_1}|} \sum_{j \,\in\, \{B,C\}} U_t^{j,m_1,m_2} d M_t^{j,\mathbb{Q}} \;, 
\end{align} 
for $m_2 = 1, \ldots, N_{\mathcal{N}}.$ \\ 
By applying Theorem \ref{fullGG:multi}, we obtain the following 
\begin{proposition} 
There exists a unique solution $(V^{m_2},Z^{m_1,m_2},U^{m_1,m_2}) \textcolor{black}{\,\in\, \cS^2(\QQ)\times\cH^{2,d}(\QQ)\times\cH^{2,2}_{\lambda}(\QQ)}$ to 
\textcolor{black}{the $\mathbb{G}$}-BSDE  
\eqref{eq:GBSDE_netting} with \textcolor{black}{close-out} condition 
\begin{align} \nonumber 
V_{\tau}^{m_2}(\varphi) & = \theta_{\tau}^{m_2} = \sum_{m_1=1}^{|\mathcal{N}_{m_2}|} \sum_{m=1}^{|\mathcal{M}_{m_1}|} 
\left(\hat{\cP}_{\tau}^{m,m_1,m_2} - DiscVA_{\tau}^{m,m_1,m_2} \right) \\ \nonumber 
& + 1_{\tau^C < \tau^B} (1 - R^C) \left(\sum_{m_1=1}^{|\mathcal{N}_{m_2}|} \sum_{m=1}^{|\mathcal{M}_{m_1}|} 
\left(\hat{\cP}_{\tau}^{m,m_1,m_2} - DiscVA_{\tau}^{m,m_1,m_2} \right) - C_{\tau-}^{\mathcal{M}_{m_1},\mathcal{N}_{m_2}} 
 - I_{\tau-}^{TC,\mathcal{M}_{m_1},\mathcal{N}_{m_2}} \right)^- \\ \label{eq:closeout_netting}
& - 1_{\tau^B < \tau^C} (1 - R^B) \left(\sum_{m_1=1}^{|\mathcal{N}_{m_2}|} \sum_{m=1}^{|\mathcal{M}_{m_1}|} 
\left(\hat{\cP}_{\tau}^{m,m_1,m_2} - DiscVA_{\tau}^{m,m_1,m_2} \right) - C_{\tau-}^{\mathcal{M}_{m_1},\mathcal{N}_{m_2}} 
 - I_{\tau-}^{FC,\mathcal{M}_{m_1},\mathcal{N}_{m_2}} \right)^+ \;.
\end{align}
\end{proposition}
\begin{remark} 
It is easy to see that the \textcolor{black}{value $\mathscr{V}^K(\varphi)$ of a} portfolio consisting of $K$ claims is \textcolor{black}{given by}
\begin{equation}\label{eq:fullPortfolioMultipleLevelXva} 
\mathscr{V}^K_t(\varphi) = \sum_{m=1}^K \hat{\cP}^m_t-\widehat{XVA}^K_t, \, \textcolor{black}{t \leq \tau} \;, 
\end{equation}
where \textcolor{black}{each $\cP^m, \, m = 1, \ldots, K,$ satisfies a BSDE of the form \eqref{ex:price}, while }
\begin{equation} \label{eq:XVA_hat_K} 
\widehat{XVA}^K_t := XVA^K_t+\sum_{m_2=1}^{N_\cN}\sum_{m_1=1}^{|\cN_{m_{2}}|} \sum_{m=1}^{|\cM_{m_{1}}|}DiscVA^{m,m_1,m_2}_t, 
\end{equation} 
and 
\begin{equation} \label{eq:XVA_K} 
XVA^K_t := FVA^K_t + ColVA^K_t + MVA^K_t - CVA^K_t + DVA^K_t \;,  
\end{equation}
\textcolor{black}{for $t \leq \tau.$} 
We further observe that each $XVA^K$ 
in \eqref{eq:XVA_K} is of the form $\sum_{m_2 = 1}^{N_{\mathcal{N}}}XVA^{m_2},$ 
where $XVA^{m_2}$ are the valuation adjustments appearing in the solution to the $\mathbb{G}$-BSDE \eqref{eq:GBSDE_netting}. 
\end{remark}

\subsection{\textcolor{black}{Incremental xVA charge}}
\textcolor{black}{
The $\GG$-BSDE \eqref{eq:GBSDE_netting} encompasses two levels of complexity. The first is the presence of several sub-portfolios, which we have discussed above. 
Secondly, the presence of bid-offer spread in the rates, coupled with the recursive 
nature of the FVA terms implies that the XVA of a portfolio of $K$ trades does not coincide with the sum of the XVAs of the single $K$ trades. This determines non-linearity effects in the xVA valuation as we describe below.}

Let us assume now that the counterparty wishes to enter into a further $(K+1)$-th trade with the bank. 
If entered, the newly introduced $(K+1)$-th claim would contribute to the global riskiness of the portfolio between the bank 
and the counterparty. It is natural to ask then what is the price the bank should charge to the newly introduced $(K+1)$-th 
claim given the presence of the already existing $K$ claims. 
\textcolor{black}{The current market practice involves the computation of an \emph{incremental xVA charge,} where	two different scenarios are compared.
			\begin{enumerate} 
			\item \textit{Base scenario:} The value of the portfolio is given by $\mathscr{V}^K(\varphi)$ as in formula 
						\eqref{eq:fullPortfolioMultipleLevelXva}. This corresponds to the value of the portfolio before the inclusion 
						of the candidate new trade.
			\item \textit{Full scenario:} The value of the portfolio is given by $\mathscr{V}^{K+1}(\varphi)$, 
						computed in line with formula \eqref{eq:fullPortfolioMultipleLevelXva}. This corresponds to the value of the portfolio 
						after the inclusion of the candidate $(K+1)$-th contingent claim.
			\end{enumerate}
			The bank determines the price to be charged to the counterparty as the difference between the value of the portfolio under 
			the full and the base scenario given by the \emph{incremental value} 
			\begin{align} \label{eq:incrementalPrice} 
			\Delta V^{K+1}_t(\varphi) := \mathscr{V}^{K+1}_t(\varphi) - \mathscr{V}^{K}_t(\varphi). 
			\end{align}
	}		 
\textcolor{black}{
We can isolate the impact of the $(K+1)$-th trade by considering  }
\begin{align}\label{eq:incrementalCharge}
\begin{aligned}
\Delta V^{K+1}_t(\varphi) &= \mathscr{V}^{K+1}_t(\varphi) - \mathscr{V}^{K}_t(\varphi)\\ 
 & = \sum_{m=1}^{K+1} \textcolor{black}{\hat{\cP}}^m_t - \widehat{XVA}^{K+1}_t - \sum_{m=1}^K \textcolor{black}{\hat{\cP}}^m_t 
+ \widehat{XVA}^K_t\\ 
 & = \textcolor{black}{\hat{\cP}}^{K+1}_t - \left(XVA^{K+1}_t-XVA^K_t\right) - DiscVA^{K+1}_t\\ 
 & = \textcolor{black}{\hat{\cP}}^{K+1}_t - \Delta XVA_t-DiscVA^{K+1}_t, 
\end{aligned}
\end{align}
\textcolor{black}{for $t \leq \tau,$}
where, in the last step, we implicitly defined the \emph{incremental xVA charge} 
\begin{equation} \label{eq:deltaXVA} 
\Delta XVA_t := XVA^{K+1}_t - XVA^K_t 
\end{equation} 
as the adjustment to be charged on the $(K+1)$-th claim, given the presence of the already existing $K$ claims in the portfolio. 
\textcolor{black}{We obtain the} \emph{non-linearity effect} defined as 
\begin{align} \label{eq:NL}
NL_t\left(V^{K+1}\right) := V^{K+1}_t(\varphi) - \Delta V^{K+1}_t(\varphi), \, \textcolor{black}{t \leq \tau} \;, 
\end{align}
where \textcolor{black}{$V^{K+1}$ is the value of the portfolio consisting only of the $(K+1)$-th claim} and $\Delta V^{K+1}_t(\varphi)$ 
is the incremental charge as defined in \textcolor{black}{\eqref{eq:incrementalPrice}}. 
The non-linearity effect coincides with the difference of the incremental xVA charge and the \textcolor{black}{xVA associated only to the $(K+1)$-th claim}, in fact 
\begin{align}
\begin{aligned}
NL_t\left(V^{K+1}\right) & = V^{K+1}_t(\varphi) - \Delta V^{K+1}_t(\varphi)\\
 & = \left(\hat{P}^{K+1}_t - XVA_t - DiscVA^{K+1}_t\right) - \left(\hat{P}^{K+1}_t - \Delta XVA_t - DiscVA^{K+1}_t\right)\\
 & = \Delta XVA_t - XVA_t. 
\end{aligned}
\end{align}
\textcolor{black}{
In the present setting the clean valuation of the contingent claim is still linear, hence the clean value of the portfolio still corresponds to the sum of the clean values of the single claims,but in general $\Delta XVA_t - XVA_t\neq 0$. The stand-alone xVA of the $(K+1)$-th claim is higher than $\Delta XVA$. 
\textcolor{black}{Moreover,} $NL_t\left(V^{K+1}\right) = 0$ only when there are no portfolio/netting effects.}

\section{\textcolor{black}{Example and numerical illustration}}\label{sec:example}
We conclude the paper by presenting an example using a lognormal model for a single risky asset. 
Under the setting and assumptions of the previous sections we consider a single risky asset $S=(S_t)_{t\in[0,T]}$ 
that pays dividends at a rate $\kappa = (\kappa_t)_{t\,\in\,[0,T]},$ with dividend process $D_t = \int_0^t\kappa_s S_s ds, \, 
t \,\in\, [0,T].$  

Under the martingale measure $\QQ$ 
the risky asset evolves according to
\begin{align} \label{eq:asset_price}
dS_t=S_t\left((r^r_t-\kappa_t)dt+\sigma_t dW^{\textcolor{black}{\QQ}}_t\right), 
\end{align}
where $r^r=\left(r^r_t\right)_{t\in[0,T]}$ is the repo rate associated to the asset $S$. 
We now consider a simple contingent claim, namely a forward written on the asset $S$. 
The dividend process of the claim $A^1=(A^1_t)_{t\in[0,T]},$ is given by
\begin{align} \label{eq:forward}
A^1_t=\Ind{t=T} (S_{T}-K^1), 
\end{align}
for $K^1$ a positive constant. We recall that the clean value $\hat{V}$ satisfying \eqref{cleanBSDE} 
represents a fictious value process for the claim under the assumption of a perfect collateralization scheme that 
annihilates counterparty risk, see Assumption~\ref{cleanMarket}. 

According to Theorem~\ref{th:cleanPriceTh} the arbitrage free price of the forward is
\begin{align}
\label{eq:XVADeskFwdCleanPrice}
\hat{V}^1_t(\varphi)=\Excond{\QQ}{B^r_t\int_{(t,T]} \frac{dA_u^1}{B^r_u}}{\cF_t} 
= B^r_t\Excond{\QQ}{\frac{S_{T}-K^1}{B^r_{T}}}{\cF_t}. 
\end{align}
Assume now that the bank enters a forward with a counterparty without any collateral agreement and without any 
previous existing trade: there is no exchange of variation or initial margin, meaning that $C=I^{TC}=I^{FC}=0$, 
$d\QQ\otimes dt$-a.s. 
Exposures on such transactions are to be funded by the internal treasury desk of the bank, hence, due to internal 
rules of the bank, the front office desk decides to discount cashflows via a synthetic unsecured discount curve with 
associated short rate process $r^f=(r^f_t)_{t\in[0,T]}$ defined via $r^f=\frac{r^{f,l}+r^{f,b}}{2}$. 

Hence, the official clean price from the bank perspective is 
\begin{align}
\label{eq:officialPrice}
\textcolor{black}{\hat{\cP}}^1_t=\Excond{\QQ}{B^f_t\int_{(t,T]} \frac{dA_u^1}{B^f_u}}{\cF_t} 
= B^f_t\Excond{\QQ}{\frac{S_{T}-K^1}{B^f_{T}}}{\cF_t}.
\end{align}
The xVA desk is forced by the internal policy of the bank to employ \eqref{eq:officialPrice} as the official clean price 
for the transaction.  However, using Proposition \ref{fullGG:multi} it is possible to compute a consistent price which is 
then given by
\begin{align*}
\mathscr{V}^{1}_t(\varphi) = V^1_t(\varphi) = \textcolor{black}{\hat{\cP}}^1_t-XVA^1_t-DiscVA^1_t, 
\end{align*}
where
\begin{align}
\label{eq:firstXVA}
\begin{aligned}
XVA^1_t&=-CVA^1_t+DVA^1_t+FVA^1_t\\
&= -B^r_t\Excond{\QQ}{
	\Ind{\tau^C<\tau^B}(1-R^C)\frac{1}{B^r_\tau}\left(\textcolor{black}{\hat{\cP}}^1_\tau - DiscVA^1_{\tau}\right)^-}{\cG_t}, \\
& +B^r_t\Excond{\QQ}{
	\Ind{\tau^B< \tau^C}(1-R^B)\frac{1}{B^r_\tau}\left(\textcolor{black}{\hat{\cP}}^1_\tau - DiscVA^1_{\tau}\right)^+}{\cG_t}, \\
&+B^r_t \Excond{\QQ}{\int_t^{\textcolor{black}{\tau}} 
	\frac{(r^{f,l}_u-r_u)\left(\mathscr{V}^1_u(\varphi)\right)^+ 
- (r^{f,b}_u-r_u)\left(\mathscr{V}^1_u(\varphi)\right)^-}{B^r_u}du}{\cG_t},
\end{aligned}
\end{align}
while the discounting adjustment is
\begin{align}
\label{eq:discVAfirstClaim}
DiscVA^1_t:=B_t^{f} \EE^{\QQ}\left[ \int_t^T (r_u - r^f_u) \frac{\hat{V}^1_u}{B_u^{f}} du \Big| \cF_t \right]. 
\end{align}

The $\GG$-BSDE solved by \eqref{eq:firstXVA} is given by
\begin{align} \label{eq:XVAGBSDEfirstClaim}
\begin{cases}
-dXVA^1_t & = -\left[f(t,\hat{V}^1-XVA^1,0,0) + r_tXVA^1_t \right]dt \\
&\quad\quad\quad- \sum_{k=1}^d\tilde{Z}^k_tdW^{k,\QQ}_t-\sum_{j\in\{B,C\}}\tilde{U}^j_tdM^{j,\QQ}_t,\\
XVA^1_\tau & = 
\hat{V}^1_\tau(\varphi)-\theta_\tau(\hat{V}^1,0,0).
\end{cases}
\end{align}
We observe that the non-linearity effect $NL_t(V^{1})=0$ is of course zero, since the portfolio between the bank and the 
counterparty consists of a single contingent claim. 

Assume now that the counterparty is interested in a second product, e.g. a second forward contract on the risky asset $S$ 
with maturity $T$ and opposite direction, so that
\begin{align}
A^2_t=\Ind{t=T}(K^2-S_{T}).
\end{align}
In line with the previous reasoning, the clean values from the perspective of the xVA desk and the front-office desk are 
respectively
\begin{align}
\hat{V}^2_t(\varphi) = B^r_t\Excond{\QQ}{\frac{K^2-S_{T}}{B^r_{T}}}{\cF_t}, 
\quad\textcolor{black}{\hat{\cP}}^2_t = B^r_t\Excond{\QQ}{\frac{K^2-S_{T}}{B^f_{T}}}{\cF_t}. 
\end{align}
Given the presence of the first forward contract in the portfolio, the full value of the portfolio, now including the second 
claim, is
\begin{align*}
\mathscr{V}^{2}_t(\varphi) = \textcolor{black}{\hat{\cP}}^1_t + \textcolor{black}{\hat{\cP}}^2_t - XVA^2_t - DiscVA^1_t - DiscVA^2_t, 
\end{align*}
where
\begin{align}
\label{eq:fullXVAWithSecond}
\begin{aligned}
XVA^2_t&=-CVA^2_t+DVA^2_t+FVA^2_t\\
&= -B^r_t\Excond{\QQ}{
	\Ind{\tau^C< \tau^B}(1-R^C)\frac{1}{B^r_\tau}\left(\textcolor{black}{\hat{\cP}}^1_\tau + \textcolor{black}{\hat{\cP}}^2_\tau - DiscVA^1_{\tau} -  DiscVA^2_{\tau}\right)^-}{\cG_t}, \\
& +B^r_t\Excond{\QQ}{
	\Ind{\tau^B< \tau^C}(1-R^B)\frac{1}{B^r_\tau}\left(\textcolor{black}{\hat{\cP}}^1_\tau + \textcolor{black}{\hat{\cP}}^2_\tau - DiscVA^1_{\tau} - DiscVA^2_{\tau}\right)^+}{\cG_t}, \\
&+B^r_t \Excond{\QQ}{\int_t^{\textcolor{black}{\tau}} 
	\frac{(r^{f,l}_u-r_u)\left(\mathscr{V}^2_u(\varphi)\right)^+ 
- (r^{f,b}_u-r_u)\left(\mathscr{V}^2_u(\varphi)\right)^-}{B^r_u}du}{\cG_t}, \\
\end{aligned}
\end{align}
and $DiscVA^2$ is given by 
\begin{align*} 
DiscVA^2_t:=B_t^{f} \EE^{\QQ}\left[ \int_t^T (r_u - r^f_u) \frac{\hat{V}^2_u}{B_u^{f}} du \Big| \cF_t \right]. 
\end{align*}

The solution to the $\GG$-BSDE 
\begin{align}
\begin{cases}
-dXVA^2_t & = -\left[f(t,\hat{V}^2-XVA^2,0,0) + r_tXVA^2_t \right]dt \\
&\quad\quad\quad- \sum_{k=1}^d\tilde{Z}^k_tdW^{k,\QQ}_t-\sum_{j\in\{B,C\}}\tilde{U}^j_tdM^{j,\QQ}_t\\
XVA^2_\tau & = 
\left(\hat{V}^1_\tau(\varphi)+\hat{V}^2_\tau(\varphi) 
- \theta_\tau(\hat{V}^1+\hat{V}^2,0,0)\right) 
\end{cases}
\end{align}
is given by \eqref{eq:firstXVA}. \\ 
Given the presence of the first claim in the portfolio, the xVA charge on the second claim is $\Delta XVA = XVA^2-XVA^1$, 
whereas the non-linearity is \textcolor{black}{effect is given by}
\begin{align}
 \label{eq:nonlinearityTwoClaims}
\begin{aligned}
NL_t(V^{2})& = XVA^2_t + B^r_t\Excond{\QQ}{
	\Ind{\tau^C< \tau^B}(1-R^C)\frac{1}{B^r_\tau} 
\left(\textcolor{black}{\hat{\cP}}^2_\tau-DiscVA^2_{\tau}\right)^-}{\cG_t}, \\
& - B^r_t\Excond{\QQ}{
	\Ind{\tau^B< \tau^C}(1-R^B)\frac{1}{B^r_\tau}\left(\textcolor{black}{\hat{\cP}}^2_\tau 
- DiscVA^2_{\tau}\right)^+}{\cG_t}, \\
& - B^r_t \Excond{\QQ}{\int_t^{\textcolor{black}{\tau}} 
	\frac{(r^{f,l}_u-r_u)\left(V^2_u(\varphi)\right)^+ 
- (r^{f,b}_u-r_u)\left(V^2_u(\varphi)\right)^-}{B^r_u}du}{\cG_t}. 
\end{aligned}
\end{align}
Observe that in the last FVA term appearing in  \eqref{eq:nonlinearityTwoClaims} we have the presence of $V^2$, 
i.e. a portfolio consisting only of the second claim: all expectations in \eqref{eq:nonlinearityTwoClaims} 
represent the stand-alone xVA correction for the second contingent claim. 
In \eqref{eq:fullXVAWithSecond} we observe instead the presence of $\mathscr{V}^2$, i.e. a portfolio consisting of 
the first and the second claim. 
The role of netting effects in reducing the overall impact of value adjustments can be seen by observing that
\begin{align}
\begin{aligned}
&\textcolor{black}{\hat{\cP}}^1_t + \textcolor{black}{\hat{\cP}}^2_t - (DiscVA^1_t+DiscVA^2_t)\\
&=\left(B^f_t\Excond{\QQ}{\frac{1}{B^f_T}}{\cF_t}-B^f_t\Excond{\QQ}{\int_t^T(r_u-r^f_u)\frac{B^r_u}{B^f_u}
\Excond{\QQ}{\frac{1}{B^r_T}}{\cF_u}du}{\cF_t}\right)(K^2-K^1)\\
&=B^r(t,T)(K^2-K^1),
\end{aligned}
\end{align}
where $B^r(t,T)$ denotes the time $t$ price of a zero coupon bond with maturity $T$ and short rate $r$. This shows that the combined exposure of the two forward contracts is obviously independent of the volatility of the asset $S.$ 

We finally stress that, given a numerical scheme that allows to estimate the evolution of the conditional expectation $\hat{V}$, 
e.g. a regression estimator in the context of a Monte Carlo simulation, the xVA desk can immediately estimate the DiscVA, 
hence only a simulation in terms of $r$ discounting is required for the implementation.

\subsection{A numerical illustration}\label{sect:numerical}
We conclude the paper by presenting two numerical illustrations\footnote{The source code for our examples is available at \url{https://github.com/AlessandroGnoatto/BiaginiEtAlExamples}}. 
The first one aims at providing evidence regarding the last claim of the previous section, namely that the estimation of DiscVA is 
a feasible task. We postulate that the risky asset evolves according to \eqref{eq:asset_price}. We assume, for the sake of simplicity, 
that all parameters are constant and we set $r^r=r = 0.01$, $\kappa=0$, $\sigma = 0.25$, $S_0 = 100$. 
In line with the previous section we consider the forward contract \eqref{eq:forward} written on $1000$ units of $S$ with strike 
$K^1=80$ and $T_1=1$. 

For such a simple claim we can compute prices without resorting to simulations. In particular, the xVA desk computes the price 
$\hat{V}^1_0$ according to \eqref{eq:XVADeskFwdCleanPrice}, i.e.  $\hat{V}^1_0=20795.22$ EUR, obtained by assuming a perfect 
collateralization with collateral rate $r$. However, the claim is perfectly uncollateralized, hence the front office function 
employs for valuation the unsecured discount rate $r^f = 0.05$ and hence computes the price $P^1_0$ according to \eqref{eq:officialPrice}. 
The front office price is $P^1_0=19980.62$ EUR. As we have already seen, the front office price is not consistent with the portfolio-wide valuation of the xVA desk, however the xVA desk can solve the consistency problem by computing $DiscVA^1_0$ as in \eqref{eq:discVAfirstClaim}. 

In this very simple example $DiscVA^1_0$ can be computed in closed form and we have  $DiscVA^1_0=-814.70$. 

Note however that \textcolor{black}{in more general cases} the computation of $DiscVA$ can be performed at no additional 
costs via Monte Carlo methods, which are usually employed for xVA simulations.

The second experiment we propose aims at showing the relevance of portfolio effects in xVA computations. Again a simple example 
will suffice to provide enough intuition. To illustrate the issue we simplify the treatment by assuming that there is a unique 
risk-free rate $r$ involved in all valuations. We consider again the bank trading two forwards on the stock $S$ as before. 
We assume again $T^1=1$ and set $K^1 = S_0=100$. We suppose that the forward is written on $1000$ units of the stock, and 
that only the counterparty can default. In summary, the whole xVA adjustment is solely given by the CVA. 
We suppose that the counterparty has a constant hazard rate $\lambda^{C,\QQ}=0.04$ and a recovery rate $R^C=0.4$.

The base scenario we consider is represented by the sole presence of the forward above in the portfolio between the bank 
and the counterparty. Using the same Monte Carlo framework we developed for the previous numerical test, we simulate again 
paths of the underlying $S$. After that, we perform a pathwise simulation of the exposure of the forward, which we then 
numerically integrate with respect to time and average over all paths. This procedure produces a Monte Carlo estimation of 
the CVA under the filtration $\FF$ according to Corollary \ref{eq:corollary}. We obtain an estimate for $CVA^1_0=148.17$ EUR.

Let us now introduce the second forward mentioned above, where we assume again that $T^2=T^1=1$ and set $K^2=90$. 
We also suppose that the second forward is written on the same quantity of shares of $S$, namely $1000$. 
We observe that, due to the different strikes, the second forward does not perfectly offset the first one. 
We first assumed that the second forward is the only claim in the portfolio and thus obtained an estimate of the 
stand alone CVA of $309.22$ EUR, so that the sum of the CVAs of the two forwards, ignoring portfolio effects, is $457.49$ EUR. 
Such value clearly overestimates the outstanding credit exposure between the bank and the counterparty.

By relying on a Monte Carlo simulation under the same assumptions as above, we compute the 
portfolio-wide CVA, i.e. 
$CVA^2$ and we obtain the estimate $CVA^2=232.69$ EUR. We observe then that the incremental CVA, $\Delta CVA = 84.52$ EUR. 
Finally, the non-linearity from \eqref{eq:NL} is estimated by $NL_0(V^2)= 232.69 - 84.52 = 148.17$ EUR. 

The example we propose shows quite clearly the relevance of portfolio effects: if the xVA desk ignored portfolio effects, 
the xVA charge would be $457.49$ EUR. By applying the incremental approach to xVA charge instead, when the second forward is 
included in the portfolio, there is only an additional charge of $84.52$ EUR. This is due to the fact that the two credit 
exposure partially compensate each other. 

In Figure \ref{fig:Exposures} we provide a further visualization of the portfolio effects. We compute the Monte Carlo sample 
average of the negative and positive part of the credit exposure of the forwards under consideration: such quantities are 
usually termed \textit{expective negative (resp. positive) exposure}. Also, we compute the $95\%$-quantile of the exposure. 
Red lines correspond to the first forward with strike $K^1$ whereas green lines refer to the second forward with strike $K^2.$ 
Finally, the portfolio resulting from the combination of the two forwards is represented by a blue line. We can clearly observe 
that combining the two claims has a beneficial effect in terms of reduction of the exposure: in particular we observe that the 
$95\%$-quantile is constant.

\begin{figure}[ht]
  \centering
   \subfloat{\label{fig:ENE}\includegraphics[scale=0.30]{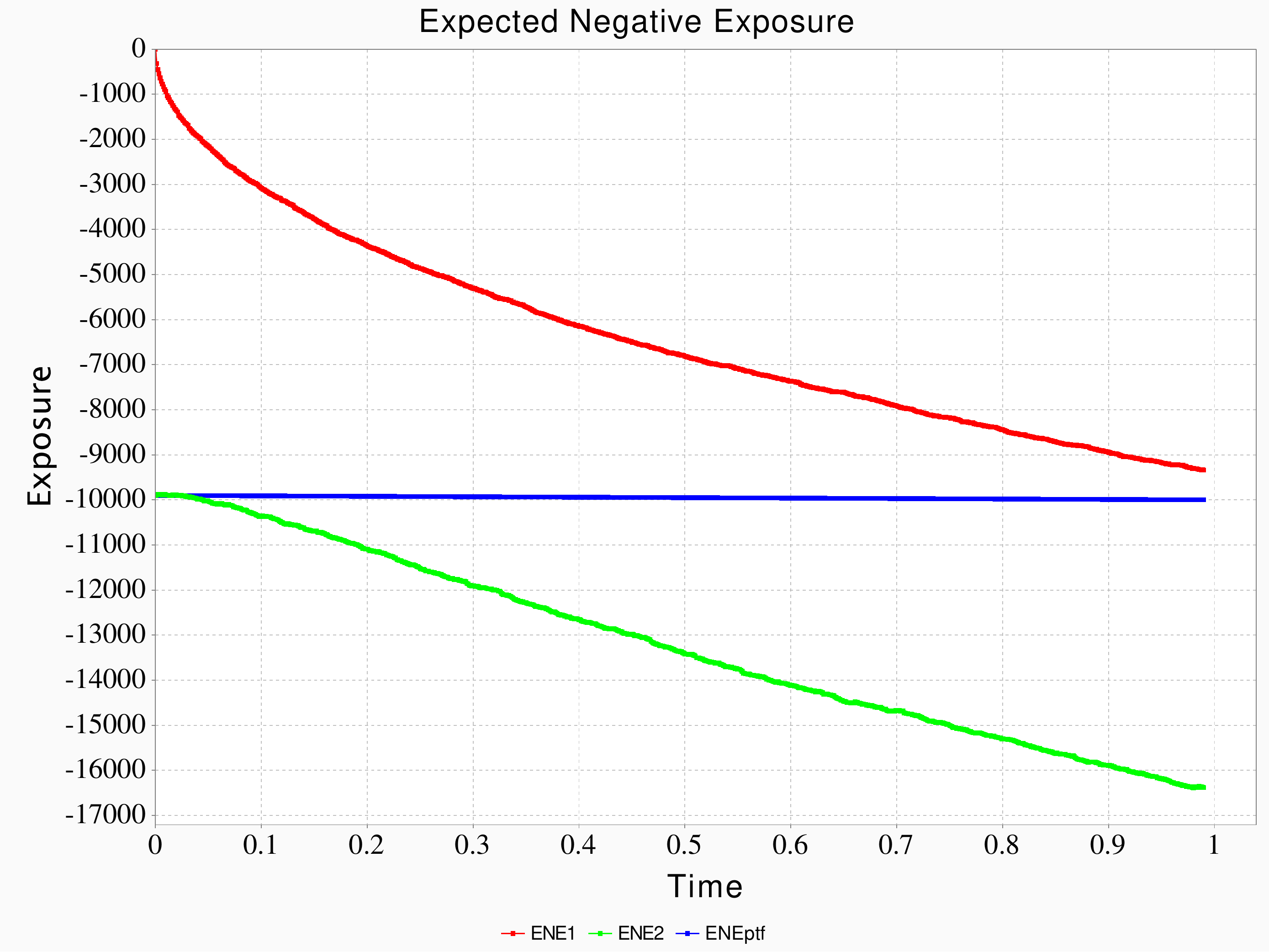}} 
   \subfloat{\label{fig:EPE}\includegraphics[scale=0.30]{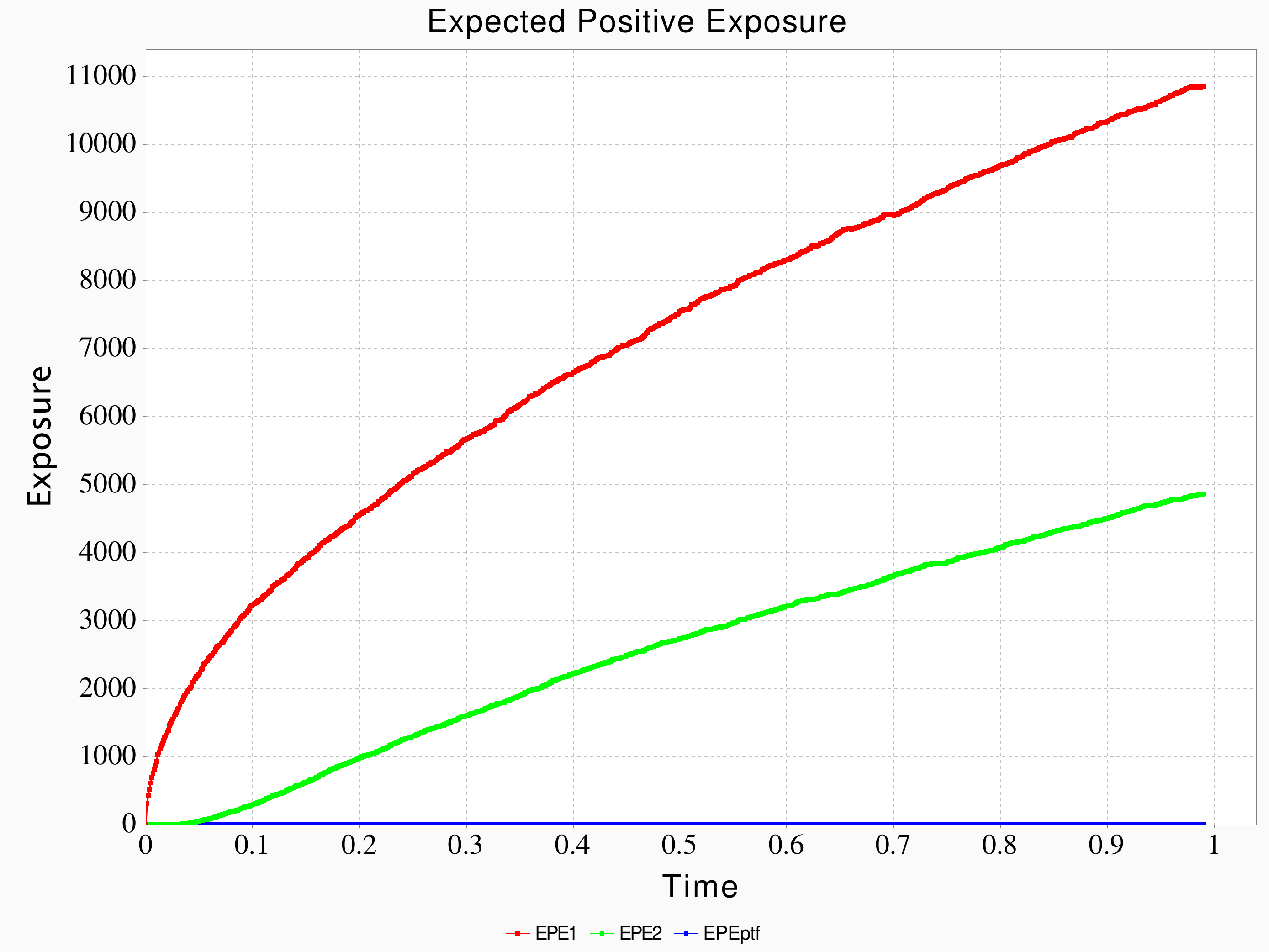}}\
    \subfloat{\label{fig:PFE}\includegraphics[scale=0.30]{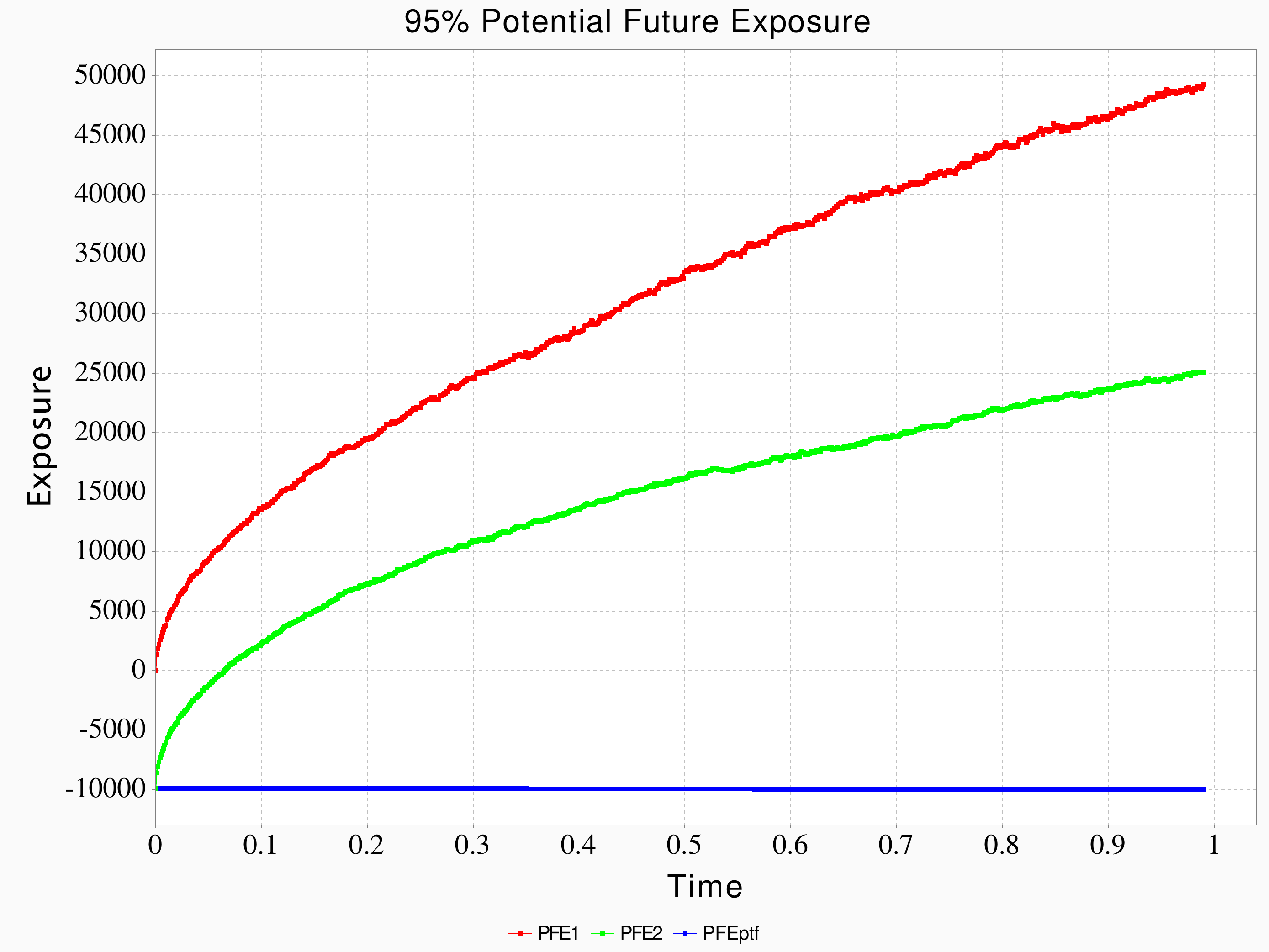}} \                              
\caption{Expected positive and negative exposure and potential future exposure for the portfolio of two forwards. The red and green lines refer to the first and second forward taken in isolation, whereas the blu line represents a portfolio containing the two claims.  \label{fig:Exposures}}
\end{figure}

\appendix
\section{Existence and Uniqueness of BSDEs}\label{app}

In this section we review some results on existence and uniqueness for some BSDEs. Our main references are \cite{nr2016}, which in turn extends results from \cite{cfs2008}, and \cite{Gobet2018}. 

Let $M = \left(M^1,\ldots M^d\right)^\top$ be a $d$-dimensional, real-valued, continuous and square integrable martingale 
on a filtered probability space $\left(\Omega,\cF,\FF,\QQ\right),$ where the filtration is assumed to satisfy the usual 
hypotheses and we assume that the predictable representation property holds with respect to $M$ for $(\FF,\QQ)$-martingales. 
We use $\left\langle M \right\rangle$ to denote the quadratic variation of $M$.

\begin{assumption}[\cite{nr2016} Assumption 3.1] There exists an $\RR^{d\times d}$-valued process $m$ and an $\FF$-adapted, 
continuous, bounded, increasing process $Q$ with $Q_0=0$ such that, for all $t\in[0,T],$ 
\begin{align}\label{eq:covariationAppendix}
\left\langle M \right\rangle_t = \int_0^t m_u m_u^\top d Q_u. 
\end{align}
\end{assumption}

If $M = W$ is a one-dimensional standard Brownian motion, then $Q_t=t$, whereas $m$ corresponds to the identity matrix. 
Next we introduce the driver of the BSDE via the following

\begin{assumption}[\cite{nr2016} Assumption 3.2]\label{assumptionDriverAppendix}
Let $h:\Omega\times [0,T]\times \RR\times \RR^d\mapsto \RR$ be an $\cF\otimes \cB([0,T])\otimes \cB(\RR)\otimes \cB(\RR^d)$-measurable function such that $h(\cdot,\cdot, y,z)$ is an $\FF$-adapted process for any fixed $(y,z)\in\RR\times \RR^d$.
\end{assumption}

The BSDEs of interest in view of financial applications are forward-backward SDEs (FBSDEs). Following \cite{nr2016}, 
we introduce a generic (forward) factor matrix-valued process given by 
\begin{align*}
\XX_t:=\left(\begin{array}{cccc}
X^1_t & 0 & \hdots & 0 \\ 
0 & X^2_t& \hdots & 0 \\ 
\vdots & \vdots & \ddots & 0\\ 
0 & 0& \hdots& X^d_t
\end{array} \right), \, t \,\in\, [0,T], 
\end{align*}
where the auxiliary processes $X^i, \, i=1,\ldots,d,$ are assumed to be $\FF$-adapted. The processes $X^i$ represent market 
risk factors or traded assets. We assume that the function $h$ of Assumption~\ref{assumptionDriverAppendix} can be written 
as $h(\omega,t,y,z)=g(\omega,t,y,\XX_tz),$ for $g$ satisfying Assumption~\ref{assumptionDriverAppendix}. 

\begin{definition}[\cite{nr2016} Definition 4.1]
We say that an $\RR^{d\times d}$-valued process $\gamma$ satisfies the ellipticity condition if there exists a constant 
$\Lambda>0$ such that
\begin{align}\label{eq:ellipticity}
\sum_{i=1}^d\left(\gamma_t\gamma_t^\top\right)_{ij}a_ia_j\geq \Lambda \left\|a\right\|
\end{align}
for all $a\in\mathbb{R}^d$ and $t\in[0,T].$
\end{definition}

\begin{assumption}[\cite{nr2016} Assumption 4.2]\label{assumptionMAppendix} The $\RR^{d\times d}$-valued $\FF$-adapted process $m$ in \eqref{eq:covariationAppendix} is given by
\begin{align*}
m_t m_t^\top=\XX_t\gamma_t\gamma_t^\top\XX_t^\top, 
\end{align*}
where $\gamma=[\gamma]_{ij}$ is a $d$-dimensional square matrix of $\FF$-adapted processes satisfying the ellipticity 
condition \eqref{eq:ellipticity}.
\end{assumption}

In the following we recall some definitions from \cite{nr2016}. 

\begin{definition}
We say that the function $h: \, \Omega\times [0,T]\times \RR\times \RR^d\mapsto \RR$ satisfies 
\begin{itemize}
\item the \emph{uniform Lipschitz condition} if there exists a constant $L$ such that for any $t\in[0,T]$ and all 
$y_1,y_2\in\RR$, $z_1,z_2\in\RR^d$
\begin{align*}
\left|h(t,y_1,z_1)-h(t,y_2,z_2)\right|\leq L\left(|y_1-y_2|+\left\|z_1-z_2\right\|\right);
\end{align*}
\item the \emph{uniform $m$-Lipschitz condition} if there exists a constant $\hat{L}$ such that for any $t\in[0,T]$ 
and all $y_1,y_2\in\RR$, $z_1,z_2\in\RR^d$
\begin{align*}
\left|h(t,y_1,z_1)-h(t,y_2,z_2)\right|\leq \hat{L}\left(|y_1-y_2|+\left\|m_t^\top(z_1-z_2)\right\|\right);
\end{align*}
\item the \emph{uniform $\XX$-Lipschitz condition} if there exists a constant $\tilde{L}$ such that for any $t\in[0,T]$ 
and all $y_1,y_2\in\RR$, $z_1,z_2\in\RR^d$
\begin{align*}
\left|h(t,y_1,z_1)-h(t,y_2,z_2)\right|\leq \tilde{L}\left(|y_1-y_2|+\left\|\XX_t(z_1-z_2)\right\|\right). 
\end{align*}
\end{itemize}
\end{definition}

\begin{lemma}[\cite{nr2016} Lemma 4.2] If Assumption~\ref{assumptionMAppendix} holds and the generator $h$ is uniform 
$\XX$-Lipschitz, then $h$ is uniform $m$-Lipschitz with $\hat{L}=\tilde{L}\max\left\{1,\Lambda^{-\frac{1}{2}}\right\}$, 
where $\Lambda$ is the constant defined in \eqref{eq:ellipticity}.
\end{lemma}

Theorem \ref{th:existenceUniquenessNR} provides the existence and uniqueness result, which is relevant for our purposes. 

\begin{theorem}[\cite{nr2016} Theorem 4.1]\label{th:existenceUniquenessNR}
Assume that the function $h$ can be represented as $h(t,y,z)=g(t,y,\XX_tz)$, where the function 
$g:\Omega\times[0,T]\times\RR\times\RR^d\mapsto\RR$ satisfies the uniform Lipschitz condition. 
Let the process $h(\cdot,0,0)$ belong to the space $\cH^{2}(\QQ)$, the random variable $\eta$ belong to $L^2(\cF_T,\QQ)$ 
and $U$ be a real-valued $\FF$-adapted process such that $U\in\cH^{2}(\QQ)$ and $U_T\in L^2(\cF_T,\QQ)$. 
Assume that the process $m$ satisfies Assumption~\ref{assumptionMAppendix} for some constant $\Lambda>0$. Then the BSDE
\begin{align}
\begin{cases}
dY_t& = Z^\top_tdM_t - h(t,Y_t,Z_t)dQ_t + dU_t,\\
Y_T&=\eta,
\end{cases}
\end{align}
has a unique solution $(Y,Z)$ such that $(Y,m^\top Z)\,\in\,\cH^{2}(\QQ) \times \cH^{2,d}(\QQ)$. 
Moreover the processes $Y$ and $U$ satisfy
\begin{align*}
\Ex{\QQ}{\sup_{t\in[0,T]}|Y_t-U_t|^2}<\infty. 
\end{align*} 
\end{theorem}

\begin{assumption}[\cite{Gobet2018} Assumption (S)] \label{ass:S_gobet}
For any $y,\,z,\,\lambda \,\in\, \RR\times\RR^d\times\RR,$ $f(\cdot,y,z,\lambda)$ is an $\FF$-adapted stochastic process with values 
in $\RR$ and there exists a constant $C_f > 0$ such that $\PP$-a.s., for all $(s,y_1,z_1,\lambda_1), \, (s,y_2,z_2,\lambda_2) 
\,\in\, [0,T]\times\RR\times\RR^d\times\RR,$ 
$$|f(s,y_1,z_1,\lambda_1) - f(s,y_2,z_2,\lambda_2)| \leq C_f \left(|y_1 - y_2| + |z_1 - z_2| + |\lambda_1 - \lambda_2| \right).$$ 
Moreover, $\EE\left[\int_0^T |f(s,0,0,0)|^2 ds \right] < \infty.$  
\end{assumption}

\subsection{Proof of Proposition \ref{pro:noArbitrageBasicMarket}} \label{app:A1}
Since we are only trading in the basic risky assets, 
the position in the initial margin is zero  hence, by \eqref{eq:repoConstraint} and \eqref{def:wealth_prox}, 
the value process is of the form 
\begin{align} \label{eq:portfolioValue}
V_t(\varphi)=\psi^{f,b}_tB^{f,b}_t+\psi^{f,l}_tB^{f,l}_t, \, t \,\in\,[0,T]. 
\end{align}
Recalling that simultaneous borrowing and lending at the same time is not allowed, we have by \eqref{eq:portfolioValue} that 
\begin{align*}
\psi^{f,l}& = \left(V_t(\varphi)\right)^+\left(B^{f,l}_t\right)^{-1}, \; 
\psi^{f,b} = -\left(V_t(\varphi)\right)^-\left(B^{f,b}_t\right)^{-1}, \; t \,\in\,[0,T].  
\end{align*}
Moreover, we can rewrite the funding term of the generic $i$-th risky assets as follows
\begin{align*}
\int_0^t\psi^i_udB^i_u=-\int_0^t\frac{\xi^i_uS^i_u}{B^i_u}dB^i_u=-\int_0^t r^i_u \xi_u^i S^i_u du, \, t \,\in\,[0,T]. 
\end{align*}
Upon substitution in the self-financing condition \eqref{eq:sefFinancing}, we obtain
\begin{align*}
\begin{aligned}
dV_t(\varphi) & = \sum_{i=1}^d\xi^i_t\left(dS^i_t + dD^i_t - r^i_tS^i_tdt\right) 
+ \sum_{j\in\{B,C\}}\xi^j_t\left(dP^j_t-r^j_tP^j_{t-}dt\right) - r^{f,b}_t\left(V_t(\varphi)\right)^-dt 
+ r^{f,l}_t\left(V_t(\varphi)\right)^+dt.
\end{aligned}
\end{align*}
We now use the inequality $r^{f,l}_t\leq r^{f,b}_t$ from Assumption~\ref{assumptionRates}, hence
\begin{align*}
\begin{aligned}
dV_t(\varphi)& = \sum_{i=1}^d\xi^i_t\left(dS^i_t+dD^i_t-r^i_tS^i_tdt\right) 
+ \sum_{j\in\{B,C\}}\xi^j_t\left(dP^j_t - r^j_tP^j_{t-}dt\right) + r^{f,l}_t\left(V_t(\varphi)\right)^+dt 
- r^{f,b}_t\left(V_t(\varphi)\right)^-dt \\ 
& \leq\sum_{i=1}^d\xi^i_t\left(dS^i_t + dD^i_t - r^{i}_tS^{i}_tdt\right) 
+ \sum_{j\in\{B,C\}}\xi^j_t\left(dP^j_t - r^{j}_tP^j_{t-}dt\right) + r^{f,l}_t V_t(\varphi)dt. \\
\end{aligned}
\end{align*}
Introducing $\tilde{V}^l_t(\varphi):=\left(B^{f,l}_t\right)^{-1}V_t(\varphi),$ we have then the inequality
\begin{align*}
\begin{aligned}
d\tilde{V}^l_t(\varphi)&\leq\sum_{i=1}^d\xi^i_t\left(B^{f,l}_t\right)^{-1}\left(dS^i_t+dD^i_t-r^{i}_tS^{i}_tdt\right) 
+ \sum_{j\in\{B,C\}}\xi^j_t\left(B^{f,l}_t\right)^{-1}\left(dP^j_t-r^{j}_tP^j_{t-}dt\right) \\
\end{aligned}
\end{align*}
or equivalently,
\begin{align*}
\begin{aligned}
d\tilde{V}^l_t(\varphi) & \leq\sum_{i=1}^d\xi^i_t\frac{B^i_t}{B^{f,l}_t}\frac{\left(dS^i_t+dD^i_t-r^{i}_tS^{i}_tdt\right)}{B^i_t} 
+ \sum_{j\in\{B,C\}}\xi^j_t\frac{B^j_t}{B^{f,l}_t}\frac{\left(dP^j_t-r^{j}_tP^j_{t-}dt\right)}{B^j_t},\\
\end{aligned}
\end{align*}
and so, by \eqref{eq:SPQ}, we arrive at the inequality
\begin{align} \label{eq:noarbitrage}
\begin{aligned}
d\tilde{V}^l_t(\varphi) &\leq \sum_{i=1}^d\xi^i_t\frac{B^i_t}{B^{f,l}_t}d\tilde{S}^{i,cld}_t + 
\sum_{j\in\{B,C\}}\xi^j_t\frac{B^j_t}{B^{f,l}_t}d\tilde{P}^j_t, \\
\end{aligned}
\end{align}
We observe that the right-hand side in \eqref{eq:noarbitrage} is a local martingale, which is bounded from below by Definition \ref{def:admissibility} \textcolor{black}{and because $r^f$ is bounded}. 
This implies that the aforementioned right-hand side is a supermartingale. Absence of arbitrage follows along the usual lines.

\section*{Acknowledgments} 
We would like to thank \textcolor{black}{the reviewers} for their useful remarks that helped us improve our work. We are also grateful to Damiano Brigo, Agostino Capponi, Andrea Pallavicini, and the participants to the SIAM Conference on Financial Mathematics \& Engineering in Toronto, the XX Quantitative Finance Workshop in Zurich, and the Mathematical and Computational Finance Internal Seminar, Oxford University, for their comments and suggestions. Any errors remain our responsibility.

\bibliographystyle{apa}
\bibliography{biblio}

\end{document}